\documentclass{tlp}

\usepackage{amsmath}
\usepackage{graphicx}
\usepackage{multirow}

\usepackage{amssymb}
\newcommand{\qed}{\hfill\ensuremath{\square}}

\usepackage{yfonts}
\usepackage{url}

\usepackage{color}

\newtheorem{lemma}{Lemma}
\newtheorem{theorem}{Theorem}
\newtheorem{corollary}{Corollary}

\newtheorem{proposition}{Proposition}

\def\G{\Gamma}

\def\seq{\Rightarrow}

\def\SQHT{\ensuremath{\mathit{SQHT}}}
\def\SQHTw{\ensuremath{\mathit{SQHT}^\omega}}
\def\Std{\ensuremath{\mathit{Std}}}
\def\Ind{\ensuremath{\mathit{Ind}}}

\def\HTC{\ensuremath{\mathit{HT}_{\!\!\#}}}
\def\HTCw{\ensuremath{\HTC^\omega}}
\def\HTCww{\ensuremath{\mathit{HT}_{\!\!\#2}^\omega}}

\def\co{\ensuremath{\mathit{count\/}}}
\def\atl{\ensuremath{\mathit{Atleast\/}}}
\def\st{\ensuremath{\mathit{Start\/}}}
\def\atm{\ensuremath{\mathit{Atmost\/}}}
\def\exa{\ensuremath{\mathit{Exactly\/}}}

\def\num{\overline}
\def\head{\ensuremath{\mathit{Head}}}
\def\body{\ensuremath{\mathit{Body}}}
\def\ar{\leftarrow}
\def\lrar{\leftrightarrow}
\def\no{\ensuremath{\mathit{not}}}

\def\gringo{{\sc gringo}}

\def\vampire{{\sc vampire}}
\def\beq{\begin{equation}}
\def\eeq#1{\label{#1}\end{equation}}
\def\ba{\begin{array}}
\def\ea{\end{array}}

\def\X{\mathcal{X}}
\def\Y{\mathcal{Y}}
\def\val#1#2{\emph{val\,}_{#1}({#2})}

\def\p2f{\hbox{p2f}}

\def\bft{{\bf t}}
\def\bfx{{\bf X}}
\def\bfxx{{\bf x}}

\def\bfyy{{\bf y}}
\def\bfz{{\bf Z}}

\def\bfu{{\bf U}}
\def\bfuu{{\bf u}}
\def\bfv{{\bf V}}
\def\bfvv{{\bf v}}
\def\bfw{{\bf W}}

\def\Defs{\ensuremath{\mathit{Defs}}}

\newcommand{\HH}{\mathcal{H}}
\newcommand{\tuple}[1]{\mbox{\ensuremath{\left\langle#1\right\rangle}}}
\newcommand{\A}{\ensuremath{\mathcal{A}}}

\newcommand{\boldd}{\mathbf{d}}

\newcommand{\boldt}{\mathbf{t}}

\newcommand{\HTi}{HT\nobreakdash-inter\-pre\-ta\-tion}
\newcommand{\HTis}{HT\nobreakdash-inter\-pre\-ta\-tions}
\newcommand{\HI}{\langle \HH,I\rangle}
\newcommand{\modelsht}{\models_{ht}}

\let\citet\cite
\let\cite\citep

\begin{document}

\lefttitle{Jorge Fandinno and Vladimir Lifschitz}

\jnlPage{1}{8}
\jnlDoiYr{2025}
\doival{10.1017/xxxxx}

\title
[Deductive Systems for Logic Programs with Counting]
{Deductive Systems for\\Logic Programs with Counting}

  
\begin{authgrp}
  \author{\sn{Jorge} \gn{Fandinno}}
  \affiliation{University of Nebraska Omaha, Omaha, NE, USA}
  \author{\sn{Vladimir} \gn{Lifschitz}}
  \affiliation{University of Texas at Austin, Austin, TX, USA}
\end{authgrp}

\newcommand{\arxivcomment}{\emph{Under consideration in Theory and Practice of Logic Programming (TPLP)}}

\maketitle

\begin{abstract}
In answer set programming, 
two groups of rules are considered strongly equivalent if they have the same meaning in any context.
Strong equivalence of two programs can be sometimes established by deriving rules of each program from rules of the other in an appropriate deductive system.
This paper shows how to extend this method of proving strong equivalence to programs containing the counting aggregate.
\arxivcomment
\end{abstract}

\section{Introduction}

In answer set programming (ASP),
two groups of rules are considered strongly equivalent if, informally speaking,
they have the same meaning in any context~\cite{lif01}. 

If programs~$\Pi_1$ and~$\Pi_2$ are strongly equivalent then,
for any program~$\Pi$, programs $\Pi_1\cup\Pi$ and $\Pi_2\cup\Pi$ have
the same stable models.
Properties of this equivalence relation are important because
they can help us simplify parts of an ASP
program without examining its other parts.  More generally, they can guide us
in the process of developing correct and efficient code.

Strong equivalence of two programs can be sometimes established by deriving
rules of each program from rules of the other in an appropriate deductive
system. Deriving rules involves rewriting them in the syntax of first-order
logic.  The possibility of such proofs has been demonstrated for the ASP
language mini\nobreakdash-\gringo\ \cite{lif19,lif21a,fan23}, and it was used
in the design
of a proof assistant for verifying strong equivalence \cite{heu20,fan23b}.

We are interested in extending this method of proving strong equivalence to
ASP programs with aggregates, such as counting and summation
\cite[Section~3.1.12]{gringomanual}.
Procedures for representing rules with aggregates in the syntax of
first-order logic have been proposed in several recent publications
\cite{lif22a,fan22a,fan23c}.  The first of these papers describes a
deductive system that can be used for proving strong equivalence of
programs in the language  called mini-\gringo\ with counting ({\sc mgc}).
But that system is too weak for reasoning about {\sc mgc} rules that
contain variables in the right-hand side of an aggregate atom.  For
instance, let~$A$ be the pair of rules
$$\ba l
p(a),\\
q(Y) \ar \co\{ X : p(X) \land X \neq a \} = Y,
\ea$$
and let~$B$ stand for
$$\ba l
p(a),\\
q(Y-1) \ar \co\{ X : p(X)\} = Y.
\ea$$
These pairs of rules are strongly equivalent to each other, but
the deductive system mentioned above would not allow us to
justify this claim.

We propose here an alternative set of axioms for proving
strong equivalence of programs with counting.  After
reviewing in Section~\ref{sec:background} the language {\sc mgc} and the
translation~$\tau^*$ that transforms {\sc mgc} rules into first-order
sentences, we define in Section~\ref{sec:htc} a deductive system of
\emph{here\nobreakdash-and\nobreakdash-there with counting (\HTC)}.
Any two {\sc mgc} programs~$\Pi_1$ and~$\Pi_2$ such that
$\tau^*\Pi_1$ and $\tau^*\Pi_2$ can be derived from each other in this
deductive system
are strongly equivalent.  Furthermore, the sentences~$\tau^*A$
and~$\tau^*B$, corresponding to the programs~$A$ and~$B$ above,
are equivalent in \HTC, as well as any two sentences that
are equivalent in the deductive system from the previous
publication on {\sc mgc} (Sections~\ref{sec:example}
and~\ref{sec:comparison}).

 The system \HTC\ is not a first-order theory in the sense of classical logic,
 because some instances of the law of excluded middle $F\lor\neg F$ are not
 provable in it.  This fact makes it difficult to automate reasoning
 in \HTC, because existing work on automated reasoning deals for the most part
 with classical logic and its extensions. 
 \citet{petowo01a} and~\citet{lin02a}
 showed how to modify the straightforward representation of propositional
 rules by formulas in such a way that strong equivalence will correspond to
 equivalence of formulas in classical logic.
 Their method
 was used in the design of a system for verifying strong equivalence
 of propositional programs \cite{che05}.  It was also generalized to
 strong equivalence of propositional formulas \cite{pea09},
 first-order formulas \cite{fer09}, and mini-\gringo\ programs
 \cite{fan23b}, and it was used in the design of a system for verifying strong
 equivalence in mini-\gringo\ \cite{heu20}.
 In Section~\ref{sec:cl} we show that this method is
 applicable to programs with counting as well.  To this end, we define
 a classical first-order theory $\HTC'$ and an additional
 syntactic transformation~$\gamma$ such that two sentences~$F_1$,~$F_2$ are
 equivalent in \HTC\ if and only if $\gamma F_1$ is equivalent to~$\gamma F_2$
 in~$\HTC'$.  It follows that if the formula
 $\gamma\tau^*\Pi_1\lrar\gamma\tau^*\Pi_2$ can be derived from the
 axioms of~$\HTC'$ in  classical first-order logic then~$\Pi_1$ is
 strongly equivalent to~$\Pi_2$.

Section~\ref{sec:omega} describes a
modificaton~$\HTCw$ of the deductive system \HTC\ that is not only
sound for proving strong equivalence, but also complete: any two {\sc mgc}
programs~$\Pi_1$,~$\Pi_2$ are strongly equivalent if and only if
the formulas $\tau^*\Pi_1$ and $\tau^*\Pi_2$ are equivalent
in the modified system.  This is achieved by including rules with
infinitely many premises, similar to the $\omega$-rule in arithmetic
investigated by Leon Henkin~(\citeyear{hen54}):
$$\frac{F(0)\quad F(1)\quad \dots}{\forall n F(n)}.$$
Deductive systems of this kind are useful as theoretical tools.
But derivations in such systems are infinite
trees, and they cannot be represented in a finite computational device.

Some additional details of the semantics of {\sc mgc} are reviewed in~\ref{sec:tau*.review}, and proofs are presented in~\ref{sec:proofs}.
Some of the proofs refer to the concept of an \HTi, which is reviewed in
Appendix~\ref{sec:review}.
In Appendix~\ref{sec:standard}, we define a class of \emph{standard} \HTis, for which the deductive system~$\HTCw$ is
sound and complete.

A preliminary report on this work was presented at the Seventeenth International Conference on Logic Programming and Non-monotonic Reasoning (LPNMR 2024).

\section{Background} \label{sec:background}

In this section, we recall the basic definitions about logic programs (Section~\ref{ssec:programs}), stable models and strong equivalence (Section~\ref{ssec:smseq}), representation of logic programs by first\nobreakdash-order formulas (Sections~\ref{ssec:programs.as.formulas} and~\ref{ssec:repag}) and the logic of here\nobreakdash-and\nobreakdash-there (Section~\ref{ssec:htsi}).

\subsection{Programs} \label{ssec:programs}

The syntax of mini-\gringo\ with counting  is
defined as follows.\footnote{The description below differs slightly from the
  original publication \cite{lif22a}:
  the absolute value symbol~$|\ |$ is allowed in the
  definition of a term, and the symbols \emph{inf} and \emph{sup} are not
  required.}
 We assume that three countably infinite sets of
symbols are selected:
\emph{numerals}, \emph{symbolic constants}, and \emph{variables}.
We assume that a 1-1 correspondence between numerals and integers is
chosen; the numeral corresponding to an integer~$n$ is denoted by~$\num n$.
(In examples of programs, we sometimes drop overlines in numerals.)

The set of \emph{precomputed terms} is assumed to be a totally ordered set containing numerals, symbolic constants, and possibly other symbols,
such that numerals are contiguous (every precomputed term between two numerals is a numeral)
and are ordered in the standard way.  {\sc mgc} terms are formed from
precomputed terms and variables using the unary operation symbol $|\ |$
and the binary operation symbols
$$+\quad-\quad\times\quad/\quad\backslash \quad..$$

An {\sc mgc} \emph{atom} is a symbolic constant optionally followed by a tuple
of terms in parentheses.  A \emph{literal} is an {\sc mgc} atom
possibly preceded by one or two occurrences of \emph{not}. A \emph{comparison}
is an expression of the form
$t_1\prec t_2$, where $t_1$, $t_2$ are mini-{\sc gringo} terms, and $\prec$ is
$=$ or one of the comparison symbols
\beq
\neq\quad<\quad>\quad\leq\quad\geq
\eeq{comp}

An \emph{aggregate element} is a pair
$\bfx:{\bf L},$
where $\bfx$ is a tuple of distinct variables,
and $\bf L$ is a conjunction of
literals and comparisons such that every member of~$\bfx$ occurs in~$\bf L$.
An \emph{aggregate atom} is an expression of one of the forms
\beq
\co \{E\} \geq t,\  \co \{E\} \leq t,\ 
\eeq{aa}
where~$E$ is an aggregate element, and~$t$ is a term that does not contain
the interval symbol ($..$).
The conjunction of aggregate atoms~(\ref{aa}) can be written as
$\co\{E\}=t$.

A \emph{rule} is an expression of the form
\beq
\head\ar\body,
\eeq{rule}
where
\begin{itemize}
\item
  \strut$\body$ is a conjunction (possibly empty) of literals, comparisons,
  and aggregate atoms, and
\item
  $\head$
  is either an atom (then~(\ref{rule}) is a
  \emph{basic rule\/}), or an atom in braces
(then~(\ref{rule}) is a \emph{choice rule}\/), or empty (then~(\ref{rule})
is a \emph{constraint}).
\end{itemize}

A variable that occurs in a rule~$R$ is \emph{local} in~$R$ if each of
its occurrences is within an aggregate element, and \emph{global} otherwise.
A rule is \emph{pure} if, for every aggregate element $\bfx:{\bf L}$ in its
body, all variables in the tuple~$\bfx$ are local.  For example,
the rule
$$
q(Y) \ar \co\{X:p(X)\} = Y \land X>0
$$
is not pure, because~$X$ is global.

In mini-\gringo\ with counting, a \emph{program} is a finite set of pure rules.

\subsection{Stable models and strong equivalence} \label{ssec:smseq}

An atom $p(\bft)$ is \emph{precomputed} if all members of the tuple~$\bft$
are precomputed terms.
The semantics of {\sc mgc} is based on an operator, called~$\tau$, which
transforms pure rules into infinitary propositional formulas formed
from precomputed atoms \cite[Section~5]{lif22a}.  For example, the rule
$$q \ar \co\{X:p(X)\} \leq 5$$
is transformed by~$\tau$ into the formula
$$\left(
\bigwedge_{\Delta\,:\,|\Delta| > 5}\neg \bigwedge_{x\in\Delta}\;p(x)
\right)\to q,$$
where $\Delta$ ranges over finite sets of precomputed terms,
and $|\Delta|$ stands for the cardinality of~$\Delta$.
The antecedent of this implication expresses that a set of more than 5 elements cannot be a subset of~$\{X:p(X)\}$.
The result of
applying~$\tau$ to a program~$\Pi$ is defined as the conjunction of
formulas $\tau R$ for all rules~$R$ of~$\Pi$.

Stable models of an {\sc mgc}
program~$\Pi$ are defined as stable models of~$\tau\Pi$ in the sense of the work by~\citet{tru12}.  Thus stable
models of programs are sets of precomputed atoms.

About programs~$\Pi_1$ and~$\Pi_2$ we say that they
are \emph{strongly equivalent} to each other if~$\tau\Pi_1$ is
strongly equivalent to~$\tau\Pi_2$; in other words, if for every set~$\Phi$
of infinitary propositional formulas formed from precomputed atoms,
$\{\tau\Pi_1\}\cup\Phi$ and $\{\tau\Pi_2\}\cup\Phi$ have the same stable
models.  It is clear that if~$\Pi_1$ is strongly equivalent
to~$\Pi_2$ then, for any program~$\Pi$, $\Pi_1\cup\Pi$ has the same stable
models as $\Pi_2\cup\Pi$ (take~$\Phi$ to be~$\{\tau\Pi\}$).

\subsection{Representing MGC terms and atoms by formulas}
\label{ssec:programs.as.formulas}

In first-order formulas used to represent programs, we distinguish between
terms of two sorts: the sort \emph{general} and its subsort
\emph{integer}.  General variables are meant to
range over arbitrary precomputed terms, and we assume them to be the same as
variables used in {\sc mgc} programs.  Integer variables are meant
to range over numerals (or, equivalently, integers).  In this paper,
integer variables are represented by letters from the middle of the
alphabet ($I,\dots,N$).

The two-sorted signature~$\sigma_0$ includes
\begin{itemize}
\item all numerals as object constants of the sort \emph{integer};
\item other precomputed terms as object constants of the sort \emph{general};
\item the symbol~$|\ |$ as a unary function constant;
  its argument and value have the sort \emph{integer};
\item the symbols~$+$, $-$ and~$\times$ as binary function constants;
  their arguments and values have the sort \emph{integer};
\item pairs $p/n$, where~$p$ is a symbolic constant and~$n$ is a nonnegative
  integer, as $n$-ary predicate constants; their arguments have the
  sort \emph{general};
\item symbols (\ref{comp}) as binary predicate constants;
     their arguments have the sort \emph{general}.
\end{itemize}
Note that the definition of~$\sigma_0$ does not allow terms that contain a
general variable in the scope of an arithmetic operation.
For example, the {\sc mgc} term $Y-\num 1$ is not a term over~$\sigma_0$.

A formula of the form $(p/n)({\bf t})$, where~$\bf t$ is a tuple of terms,
can be written also as~$p({\bf t})$.
Thus precomputed atoms can be viewed as sentences over~$\sigma_0$.

The set of values of an {\sc mgc} term\footnote{We talk
  about a \emph{set} of values because an {\sc mgc} term may contain the
  interval symbol.  For instance, the values of the {\sc mgc} term
  $\num 1 .. \num 3$ are $\num 1$,
  $\num 2$, and $\num 3$.  On the other hand, the set of values of
  the term $a-\num 1$, where~$a$ is a symbolic constant, is empty.}~$t$
can be described by a formula over the signature~$\sigma_0$
that contains a variable~$Z$ that does not occur in~$t$
\cite{lif19,fan23}.  This formula, ``$Z$ is a value of~$t$,''
is denoted by~$\val t Z$.
Its definition is recursive, and we reproduce here two of its clauses:
\begin{itemize}
\item
if $t$ is a precomputed term or a variable then $\val t Z$ is $Z=t$,
\item
if $t$ is $t_1\;\emph{op}\;t_2$, where \emph{op} is $+$, $-$, or $\times$
then $\val t Z$ is
$$\exists I J  (\val{t_1} I \land \val{t_2} J \land Z=I\;\emph{op}\;J),$$
where~$I$ and~$J$ are integer variables that do not occur in~$t$.
\end{itemize}
For example, $\val{Y-\num 1} Z$ is
$$\exists I J(I=Y \land J=\num 1\land Z=I-J).$$

The translation~$\tau^B$ transforms {\sc mgc} atoms, literals and
comparisons into formulas over the signature~$\sigma_0$.  (The
superscript~$B$ conveys the idea that this translation expresses the
meaning of expressions in \emph{bodies} of rules.)  For example,~$\tau^B$
transforms $p(t)$ into the formula $\exists Z(\val t Z \land p(Z))$.
The complete definition of~$\tau^B$ can be found in earlier publications
\cite{lif19,fan23} and is reproduced in~\ref{sec:tau*.review} for convenience.

\subsection{Representing aggregate expressions and rules} \label{ssec:repag}

To represent aggregate expressions by first-order formulas, we need to
extend the signature~$\sigma_0$ \cite[Section~7]{lif22a}.
The signature~$\sigma_1$ is obtained from~$\sigma_0$ by adding all
predicate constants of the forms
\beq
\atl^{\bfx;\bfv}_F\hbox{ and }\atm^{\bfx;\bfv}_F
\eeq{newpc}
where $\bfx$ and $\bfv$
are disjoint lists of distinct general variables, and~$F$ is a formula
over~$\sigma_0$ such that each of its free variables belongs to $\bfx$
or to~$\bfv$. The number of arguments of each constant~(\ref{newpc})
is the length of~$\bfv$ plus 1; all arguments are of the sort \emph{general}.
If~$n$ is a positive integer then the
formula $\atl^{\bfx,\bfv}_F(\bfv,\num n)$ is meant to express
that~$F$ holds for at least~$n$ values of~$\bfx$.  The intuitive meaning
of $\atm^{\bfx,\bfv}_F(\bfv,\num n)$ is similar:
$F$ holds for at most~$n$ values of~$\bfx$.

For an aggregate atom of the form $\co\{\bfx:{\bf L}\}\geq t$ in the body of
a rule, the corresponding formula over~$\sigma_1$ is
$$
\exists Z\left(\val t Z
\land \atl^{\bfx;\bfv}_{\exists \bfw\tau^B({\bf L})}(\bfv,Z)\right),
$$
where
\begin{itemize}
\item $\bfv$ is the list of global variables that occur in ${\bf L}$, and
\item $\bfw$ is the list of local variables that occur in ${\bf L}$ and
  are different from the members of~$\bfx$.
\end{itemize}
For example, the aggregate atom
$\co\{ X : p(X)\} \geq Y$ is represented by the formula
$$
\exists Z\left(Z=Y\land \atl^{X;}_{\exists Z(Z=X\land p(Z))}(Z)\right)
$$
($\bfv$ and $\bfw$ are empty).

The formula representing $\co\{\bfx:{\bf L}\}\leq t$ is formed in a similar
way, with~$\atm$ in place of~$\atl$.

Now we are ready to define the translation~$\tau^*$, which transforms
pure rules into sentences over~$\sigma_1$.  It converts a basic rule
$$p(t) \ar B_1\land\cdots\land B_n$$
into the universal closure of the formula
$$B^*_1\land\cdots\land B^*_n\land\val t Z\to p(Z),$$
where~$B^*_i$ is
\begin{itemize}
\item $\tau^B(B_i)$, if~$B_i$ is a literal or comparison, and
\item the formula representation of~$B_i$ formed as described above,
  if~$B_i$ is an aggregate atom.
\end{itemize}

The definition of~$\tau^*$ for pure rules of other forms can be found
in the previous paper on mini-\gringo\ with counting
\cite[Sections~6~and~8]{lif22a}.  For any program~$\Pi$, $\tau^*\Pi$
stands for the conjunction of the
sentences~$\tau^*R$ for all rules~$R$ of~$\Pi$.

\subsection{Logic of here\nobreakdash-and\nobreakdash-there and standard interpretations}
\label{ssec:htsi}

We are interested in deductive systems~$S$ with the following property:
\beq\ba c
\hbox{\emph{for any programs~$\Pi_1$ and~$\Pi_2$}},\\
\hbox{\emph{if $\tau^*\Pi_1$ and $\tau^*\Pi_2$ can be derived from each
    other in~$S$}}\\
\hbox{\emph{then $\Pi_1$ is strongly equivalent to~$\Pi_2$.}}
\ea\eeq{p}

Systems with property~(\ref{p})
cannot possibly sanction unlimited use of classical
propositional logic.  Consider, for instance, the one-rule programs
$$p \ar\no\ q\quad\hbox{and}\quad q \ar\no\ p.$$
They have different stable models, although the corresponding formulas
$$\neg q\to p\quad\hbox{and}\quad\neg p\to q$$
have the same truth table.

This observation suggests that the study of subsystems of
classical logic may be relevant.  One such subsystem is first-order
intuitionistic logic (with equality)  adapted to the two-sorted
signature~$\sigma_1$.  Intuitionistic logic does have property~(\ref{p}).
Furthermore, this property is preserved if we extend it by the axiom schema
\beq
F\lor (F\to G) \lor \neg G,
\eeq{hosoi}
introduced by \citet{hos66} as part of his
formalization of the propositional logic known as the logic of
here\nobreakdash-and\nobreakdash-there.

The axiom schema
\beq
\exists X(F\to \forall X\,F)
\eeq{sqht}
(for a variable~$X$ of either sort)
can be included without losing property~(\ref{p}) as well.
It was introduced to extend the logic
of here\nobreakdash-and\nobreakdash-there to a language with variables and quantifiers \cite{lif07a}.
Both~(\ref{hosoi}) and~(\ref{sqht}) are provable classically, but not
intuitionistically.

The axioms and inference rules discussed so far are abstract, in the
sense that they are not related to any
properties of the domains of variables (except that one is a subset of the
other).  To describe more specific axioms, we need the following definition.
An interpretation of the signature~$\sigma_0$ is \emph{standard} if
 \begin{itemize}
 \item its domain of the sort \emph{general} is the set of
   precomputed terms;
 \item its domain of the sort \emph{integer} is the set of numerals;
 \item every object  constant represents itself;
 \item the absolute value symbol and the binary function constants are
   interpreted as usual in arithmetic;
 \item predicate constants~(\ref{comp}) are interpreted in accordance
   with the total order on precomputed terms chosen in the definition
   of {\sc mgc} (Section~\ref{ssec:programs}).
 \end{itemize}
 Two standard interpretations of~$\sigma_0$ can differ only by how they
 interpret the predicate symbols~$p/n$.  If a sentence over~$\sigma_0$
 does not contain these symbols then it is either satisfied by all
 standard interpretations or is not satisfied by any of them.

Let \Std\  be the set of sentences over~$\sigma_0$ that do not
contain predicate symbols of the form~$p/n$ and are satisfied by
standard interpretations.   
Property~(\ref{p}) will be preserved if we add
any members of \Std\  to the set of axioms.  The set \Std\  includes,
for instance, the law of excluded middle $F\lor\neg F$ for every formula~$F$
over~$\sigma_0$ that does not contain symbols~$p/n$.  Other examples of
formulas from \Std\  are
$$\num 2\times\num 2 = \num 4,\quad
\forall N (N*N\geq\num 0),\quad t_1\neq t_2,
$$
where $t_1$, $t_2$ are distinct precomputed terms.

To reason about {\sc mgc} programs, we need also axioms for~$\atl$
and~$\atm$. A possible choice of such additional axioms
is described in the next section.

\section{Deductive system \HTC} \label{sec:htc}

The deductive system \HTC\  (``here\nobreakdash-and\nobreakdash-there with counting'') operates
with formulas of the signature~$\sigma_2$, which is obtained
from~$\sigma_1$ (Section~\ref{ssec:repag}) by adding  the
predicate constants $\st^{\bfx;\bfv}_F$,
where~$\bfx$ and~$\bfv$ are disjoint lists of distinct general variables,
and~$F$ is a formula
over~$\sigma_0$ such that each of its free variables belongs to $\bfx$
or to~$\bfv$.  The number of arguments of each of these constants
is the combined length of~$\bfx$ and~$\bfv$ plus~1.  The last argument is of
the sort
\emph{integer}, and the other arguments are of the sort \emph{general}.

For any integer~$n$, the formula $\st^{\bfx;\bfv}_F(\bfx,\bfv,\num n)$ is
meant to express that if $n>0$ then
there exists a lexicographically increasing sequence~$\bfx_1,\dots\bfx_n$
of values satisfying~$F$ such that the first of them is~$\bfx$.
Two features of the~$\st$ predicate make it useful.
On the one hand, it can be described by a recursive definition.
On the other hand, the predicates~$\atl$ and~$\atm$ can be defined in terms of~$\st$ as illustrated by~formulas~\eqref{d1a} and~\eqref{d1b} below.

\subsection{Axioms of \HTC}  \label{ssec:htc}

The axioms for~$\st$ define these predicates recursively:
$$
\ba l
\forall\bfx\bfv N(N\leq \num 0 \to \st^{\bfx;\bfv}_F(\bfx,\bfv,N)),\\
\forall\bfx\bfv (\st^{\bfx;\bfv}_F(\bfx,\bfv,\num 1)\lrar F),\\
\forall\bfx\bfv N(N> \num 0 \to (\st^{\bfx;\bfv}_F(\bfx,\bfv,N+\num 1)
\lrar\\ \hskip 4cm
F\land\exists \bfu(\bfx<\bfu\land \st^{\bfx;\bfv}_F(\bfu,\bfv,N)))).
\ea$$
Here~$N$ is an integer variable, and~$\bfu$ is a list of distinct
general variables of the same length as~$\bfx$, which is disjoint from
both~$\bfx$ and~$\bfv$.  The symbol~$<$ in the last line denotes
lexicographic order: $(X_1,\dots,X_m)<(U_1,\dots,U_m)$ stands for
$$
\bigvee_{l=1}^m\left((X_l<U_l) \land \bigwedge_{k=1}^{l-1}(X_k=U_k)\right).
$$
This set of axioms for~$\st$ will be denoted by~$D_0$.

The set of axioms for~$\atl$ and~$\atm$,
denoted by~$D_1$, defines these predicates in terms of~$\st$:
\begin{align}
\forall \bfv Y(\atl^{\bfx;\bfv}_F(\bfv,Y)
&\lrar
\exists \bfx N(\st^{\bfx;\bfv}_F(\bfx,\bfv,N)\land N\geq Y)),
\label{d1a}
\\
\forall \bfv Y(\atm^{\bfx;\bfv}_F(\bfv,Y)
&\lrar
\forall \bfx N(\st^{\bfx;\bfv}_F(\bfx,\bfv,N)\to N\leq Y)).
\label{d1b}
\end{align}

In addition to the axioms listed above, we need the induction schema%
\footnote{If~$F$ is a formula, by~$F^X_t$ we denote the result of
  replacing all free occurrences of a variable~$X$ in~$F$ by a term~$t$.}
$$
F^N_{\num 0}\land \forall N\left(N\geq\num 0 \land F \to F^N_{N+\num 1}\right)
\to \forall N (N\geq\num 0 \to F)
$$
for all formulas $F$ over~$\sigma_2$.  The set of the universal closures
of its instances will be denoted by \Ind.

The deductive system \HTC\  is defined as first-order intuitionistic logic
for the signature~$\sigma_2$ extended by
\begin{itemize}
\item axiom schemas (\ref{hosoi}) and (\ref{sqht}) 
  for all formulas $F$, $G$, $H$ over~$\sigma_2$, and
\item
axioms \Std, \Ind, $D_0$ and $D_1$.
\end{itemize}

This deductive system has property~(\ref{p}):
\begin{theorem}
  \label{thm:j:HTC.strong.equivalence}
  For any programs~$\Pi_1$ and~$\Pi_2$, if the formula
  $\tau^*\Pi_1\lrar\tau^*\Pi_2$
  is provable in \HTC\ then~$\Pi_1$ and~$\Pi_2$ are
  strongly equivalent.
\end{theorem}

The proof of this theorem can be found in Appendix~\ref{sec:proof:thm:j:HTC.strong.equivalence} (page~\pageref{sec:proof:thm:j:HTC.strong.equivalence}).
As discussed in Section~\ref{sec:omega},
\HTC\ is an extension of the system with
property~(\ref{p}) introduced by \citet{lif22a}.
Furthermore, in Section~\ref{sec:example} we show that \HTC\  is
sufficiently strong
for proving the equivalence between~$\tau^*A$ and~$\tau^*B$ for the
programs~$A$ and~$B$ from the introduction.

\subsection{Some theorems of \HTC}

The characterization of~$\atl$ and~$\atm$ given by the axioms~$D_1$ can be
simplified, if we replace the variable~$Y$ by an integer variable:

\begin{proposition}\label{prop:simp}
  The formulas
\beq
\forall \bfv N(\emph{\atl}^{\bfx;\bfv}_F(\bfv,N)
\lrar
\exists\bfx\,\emph{\st}^{\bfx;\bfv}_F(\bfx,\bfv,N))
\eeq{htc:atln}
and
\beq
\forall \bfv N(\emph{\atm}^{\bfx;\bfv}_F(\bfv,N)
\lrar
\neg\emph{\atl}^{\bfx;\bfv}_F(\bfv,N+\num 1))
\eeq{htc:atmn1}
are provable in \HTC.
\end{proposition}

A few other theorems of \HTC:

\begin{proposition}\label{prop:other}
The formulas
\beq
\forall \bfv N(N\leq \num 0\to \emph{\atl}^{\bfx;\bfv}_F(\bfv,N)),
\eeq{atlnonp}
\vskip -5mm

\beq
\forall \bfv (\emph{\atl}^{\bfx;\bfv}_F(\bfv,\num1) \lrar \exists\bfx\,F),
\eeq{atl1}
\beq
\forall \bfx(F\to G)\to
\forall\bfx\bfv N(\emph{\st}^{\bfx;\bfv}_F(\bfx,\bfv,N)
\to\emph{\st}^{\bfx;\bfv}_G(\bfx,\bfv,N)),
\eeq{mon4a}
\beq
\forall \bfz\bfv N (\emph{\st}^{\bfx;\bfv}_F(\bfz,\bfv,N)\land N>\num 0
\to F^\bfx_\bfz)
\eeq{stpos}
are provable in \HTC.
\end{proposition}

An expression of the form $\exa^{\bfx;\bfv}_F(\bft,t)$ is shorthand for the
conjunction
$$\atl^{\bfx;\bfv}_F(\bft,t)\land\atm^{\bfx;\bfv}_F(\bft,t)$$
($\bft$ is a tuple of terms, and~$t$ is a term).
By~(\ref{htc:atmn1}), $\exa^{\bfx;\bfv}_F(\bfx,N)$ is equivalent
in~\HTC\  to
$$\atl^{\bfx;\bfv}_F(\bfx,N)\land\neg \atl^{\bfx;\bfv}_F(\bfx,N+\num 1).$$

\begin{proposition}\label{prop:exactly1}
  The formulas
\beq
\forall\bfx Y(\emph{\exa}^{\bfx;\bfv}_F(\bfx,Y)
\to\exists N(Y=N\land N\geq\num 0))
\eeq{exan}
and
\beq
\forall \bfx(F\lrar G)\to
\forall\bfx Y(\emph{\exa}^{\bfx;\bfv}_F(\bfx,Y)\lrar
\emph{\exa}^{\bfx;\bfv}_G(\bfx,Y))
\eeq{mon4d}
are provable in \HTC.
\end{proposition}

The proof of these results can be found in Appendix~\ref{sec:proofs-props} (page~\pageref{sec:proofs-props}).

\section{An example of reasoning about programs} \label{sec:example}

In this section,
we show that $\tau^*A$ is equivalent to $\tau^*B$, for
the programs~$A$ and~$B$ from the introduction, in the
logic of here\nobreakdash-and\nobreakdash-there with counting, defined in the previous section.
The proof consists of three parts.

\subsection{Part 1: Simplification}

The translation~$\tau^*$ transforms program~$A$ into the conjunction of the
formulas
\beq
\forall Z(Z=a\to p(Z))
\eeq{f00}
and
\begin{align}
  \forall Y\!Z\big(\exists Z_1(Z_1\!=\!Y\!\land\!\atl^{X;}_F(Z_1))\!\land\!
  \exists Z_2(Z_2\!=\!Y\!\land\!\atm^{X;}_F(Z_2))\!\land\! Z\!=\!Y
  \to q(Z)\big),  
  \label{f000}
\end{align}
where~$F$ stands for $\tau^B(p(a)\land X\neq a)$.  Formula~(\ref{f00})
is equivalent to~$p(a)$, and~(\ref{f000}) is equivalent to
\beq
\forall Y(\atl^{X;}_F(Y)\land \atm^{X;}_F(Y)\to q(Y)).
\eeq{f0}
The antecedent of this implication can be written as $\exa^{X;}_F(Y)$.
By~(\ref{exan}), it follows that the variable~$Y$ can be replaced
by the integer variable~$N$.  Furthermore, by~(\ref{mon4d}),
formula~(\ref{f0}) can be further rewritten as
\beq
\forall N(\exa^{X;}_{p(a)\land X\neq a}(N)\to q(N)),
\eeq{f1}
because $F$ is equivalent to $p(a)\land X\neq a$.

The result of applying~$\tau^*$ to~$B$ is
the conjunction of~(\ref{f00})
and
\begin{gather*}
  \begin{aligned}
    \forall YZ\big(
    &\exists Z_1(Z_1=Y\land\atl^{X;}_G(Z_1))\,\land\qquad\qquad\\
    &\exists Z_2(Z_2=Y\land\atm^{X;}_G(Z_2))\,\land\qquad\;\;\\
    &\exists IJ(I=Y\land J=\num 1\land Z=I+J)
    \to q(Z)\big),    
  \end{aligned}
\end{gather*}
where~$G$ stands for $\tau^B(p(X))$.
This formula can be equivalently rewritten as
$$\forall I(\exa^{X;}_G(I+\num 1)\to q(I))$$
and further as
\beq
\forall I(\exa^{X;}_{p(X)}(I+\num 1)\to q(I)),
\eeq{f2}
because $G$ is equivalent to $p(X)$.

Thus the claim that $\tau^*A$ is equivalent to~$\tau^*B$
will be proved if we prove the formula
$$
p(a) \to
\forall N(\exa^{X;}_{p(X)\land X\neq a}(N)\lrar\exa^{X;}_{p(X)}(N+\num 1)).
$$
It is clearly a consequence of the formula
\beq
p(a) \to
\forall N(\atl^{X;}_{p(X)\land X\neq a}(N)\lrar\atl^{X;}_{p(X)}(N+\num 1)),
\eeq{l1}
which is proven below.

\subsection{Part 2: Three lemmas}

Three lemmas will be proved in the next section:
\begin{align}
&\forall XN(N>\num 0\land X > a\land \st^{X;}_{p(X)}(X,N)
\to \st^{X;}_{p(X)\land  X\neq a}(X,N)),\label{l3half}\\
&\forall XN(N>\num 0\land X\neq  a\land
\st^{X;}_{p(X)}(X,N+\num 1)\to\st^{X;}_{p(X)\land X\neq a}(X,N)),\label{l4lr}
\end{align}
and
\beq\ba r
\forall XN(N>\num 0\land X<a\land p(a)\land
\st^{X;}_{p(X)\land X\neq a}(X,N)\to
\st^{X;}_{p(X)}(X,N+\num 1)).
\ea\eeq{l4rl}
Using these lemmas, we will now prove~(\ref{l1}).
Assume $p(a)$; our goal is to show that
$$
\atl^{X;}_{p(X)\land X\neq a}(N)\lrar\atl^{X;}_{p(X)}(N+\num 1).
$$
We consider three cases, according to the axiom
$$\forall N(N<\num 0 \lor N=\num 0 \lor N>\num 0)$$
from \Std\ .

If~$N<\num 0$ then both sides of the equivalence are true
by~(\ref{atlnonp}).  If $N=\num 0$ then the left-hand
side is true by~(\ref{atlnonp}), and the right-hand side
follows from $p(a)$ by~(\ref{atl1}).
Hence, assume that $N>\num 0$.

\medskip\noindent\emph{Right-to-left:} assume $\atl^{X;}_{p(X)}(N+\num 1)$.
By (\ref{htc:atln}), there exists~$X$ such that
\beq
\st^{X;}_{p(X)}(X,N+\num 1).
\eeq{int9}
\emph{Case~1:} $X=a$, so that $\st^{X;}_{p(X)}(a,N+\num 1)$.
By $D_0$,
$$p(a)\land\exists U(a<U\land\st^{X;}_{p(X)}(U,N)).$$
Take~$U$ such that $a<U$ and $\st^{X;}_{p(X)}(U,N)$.
By~(\ref{l3half}), it follows that
$$\st^{X;}_{p(X)\land X\neq a}(U,N).$$  Then
$\atl^{X;}_{p(X)\land X\neq a}(N)$ by (\ref{htc:atln}).
\\[10pt]
\emph{Case~2:} $X\neq a$.  By~(\ref{int9}) and~(\ref{l4lr}),
$$\st^{X;}_{p(X)\land X\neq a}(X,N).$$
By (\ref{htc:atln}), it follows that
$\atl^{X;}_{p(X)\land X\neq a}(N)$.

\medskip\noindent\emph{Left-to-right:}
assume $\atl^{X;}_{p(X)\land X\neq a}(N)$.
Then, for some~$X$,
\beq
\st^{X;}_{p(X)\land X\neq a}(X,N)
\eeq{int1}
by (\ref{htc:atln}), and consequently $\st^{X;}_{p(X)}(X,N)$ by~(\ref{mon4a}).
\\[10pt]
\emph{Case~1:} $X>a$.  Then
$$p(a)\land \exists U(a< U\land \st^{X;}_{p(X)}(U,N))$$
(take~$U$ to be~$X$).  By $D_0$, we can conclude that
$\st^{X;}_{p(X)}(a,N+\num 1)$.  Then ${\atl^{X;}_{p(X)}(N+\num 1)}$
follows by (\ref{htc:atln}).
\\[10pt]
\emph{Case~2:} $X\leq a$. From~(\ref{int1}) and~(\ref{stpos}),
\hbox{$X\neq a$}, so that $X<a$. From~(\ref{int1}) and~(\ref{l4rl}),
$\st^{X;}_{p(X)}(X,N+\num 1)$; $\atl^{X;}_{p(X)}(N+\num 1)$
follows by (\ref{htc:atln}).

\subsection{Part 3: Proofs of the lemmas}

Proofs of all three lemmas use induction in the form
\beq
F^N_{\,\num 1}\land \forall N\left(N\geq\num 1 \land F \to F^N_{N+\num 1}\right)
\to \forall N (N\geq\num 1 \to F),
\eeq{ind1}
which follows from \Ind\ and \Std.

\medskip\noindent\emph{Proof of~(\ref{l3half}).}
We need to show that for all positive~$N$,
\beq\forall X(X>a \land\st^{X;}_{p(X)}(X,N)\to
\st^{X;}_{p(X)\land X\neq a}(X,N)).
\eeq{ih}
If $N$ is~$\num 1$ then~(\ref{ih}) is equivalent to
$$\forall X(X>a \land p(X) \to p(X)\land X\neq a)$$
by~$D_0$; this formula follows from \Std.
Assume~(\ref{ih}) for a positive~$N$; we need to prove
$$\forall X(X>a \land\st^{X;}_{p(X)}(X,N+\num 1)\to
\st^{X;}_{p(X)\land X\neq a}(X,N+\num 1)).$$
Assume $X>a \land\st^{X;}_{p(X)}(X,N+\num 1)$.
From the second conjunctive term,
$$
p(X)\land \exists U(X<U\land \st^{X;}_{p(X)}(U,N))
$$
by $D_0$.  Take~$U$ such that $X<U$ and $\st^{X;}_{p(X)}(U,N)$.
Then $U>a$, so that by the induction hypothesis,
$\st^{X;}_{p(X)\land X\neq a}(U,N)$.
Since $p(X)$, $X\neq a$, and $X<U$,
$$\st^{X;}_{p(X)\land X\neq a}(X,N+\num 1))$$
follows by $D_0$.

\medskip\noindent\emph{Proof of~(\ref{l4lr}).}
We need to show that for all positive~$N$,
\beq
\forall X(X\neq  a\land
\st^{X;}_{p(X)}(X,N+\num 1)\to\st^{X;}_{p(X)\land X\neq a}(X,N)).
\eeq{ihh}
To prove this formula for $N$ equal to $\num 1$,
assume that $X\neq a\land \st^{X;}_{p(X)}(X,\num 2)$.  By~(\ref{stpos}),
the second conjunctive term implies~$p(X)$;
$\st^{X;}_{p(X)\land X\neq a}(X,\num 1)$ follows by $D_0$.
Now assume~(\ref{ihh}) for a positive~$N$; we need to prove
\beq
\forall X(X\neq a\land
\st^{X;}_{p(X)}(X,N+\num 2)\to\st^{X;}_{p(X)\land X\neq a}(X,N+\num 1)).
\eeq{g}
Assume $X\neq a\land \st^{X;}_{p(X)}(X,N+\num 2)$.
From the second conjunctive term we conclude  by~$D_0$
that~$p(X)$ and, for some~$U$,
\beq
U>X\land \st^{X;}_{p(X)}(U,N+\num 1).
\eeq{int7}
We proceed by cases.
\\[10pt]
\emph{Case~1:} $U=a$, so that
$\st^{X;}_{p(X)}(a,N+\num 1)$.  By $D_0$, it follows that for some~$V$,
$V>a\land \st^{X;}_{p(X)}(V,N)$.  Then, by~(\ref{l3half}),
$\st^{X;}_{p(X)\land X\neq a}(V,N)$.  On the other hand, 
$p(X)\land X\neq a$ and \hbox{$V>a=U>X$};
the consequent of~(\ref{g}) follows by $D_0$.
\\[10pt]
\emph{Case~2:} $U\neq a$.  By the induction hypothesis, from the second
conjunctive term of~(\ref{int7}) we conclude that
$\st^{X;}_{p(X)\land X\neq a}(U,N)$.  Since \hbox{$U>X$} and
$p(X)\land X\neq a$,
the consequent of~(\ref{g}) follows by $D_0$.

\medskip\noindent\emph{Proof of~(\ref{l4rl}).}
We need to show that for all positive~$N$,
\beq
\forall X(X<a \land p(a)\land
\st^{X;}_{p(X)\land X\neq a}(X,N)\to \st^{X;}_{p(X)}(X,N+\num 1)).
\eeq{ihhh}
To prove this formula for $N$ equal to $\num 1$, assume that
$$X<a\land p(a)\land\st^{X;}_{p(X)\land X\neq a}(X,\num 1).$$
By $D_0$,
the second conjunctive term implies $\st^{X;}_{p(X)}(a,\num 1)$, and
the third term implies~$p(X)$.  Hence
$$p(X)\land\exists U(X<U \land \st^{X;}_{p(X)}(U,\num 1))$$
(take~$U$ to be~$a$). By $D_0$, it follows that
$\st^{X;}_{p(X)}(X,\num 2)$.
Now assume~(\ref{ihhh}) for a positive~$N$; we need to prove
$$\forall X(X<a\land p(a)\land
\st^{X;}_{p(X)\land X\neq a}(X,N+\num 1)\to \st^{X;}_{p(X)}(X,N+\num 2)).$$
Assume $X<a\land p(a)\land \st^{X;}_{p(X)\land X\neq a}(X,N+\num 1)$.
From the last conjunctive term we conclude by $D_0$\ that $p(X)$
and there exists~$U$ such that
\beq
X<U\land\st^{X;}_{p(X)\land X\neq a}(U,N).
\eeq{int5}
From the second conjunctive term of~(\ref{int5}), by~(\ref{stpos}),
$p(U)$ and~$U\neq a$.
We proceed by cases.
\\[10pt]
\emph{Case~1:} $U<a$.  By the induction hypothesis,
$\st^{X;}_{p(X)}(U,N+\num 1)$.  Since $p(X)$ and $X<U$, we can conclude
by $D_0$\ that $\st^{X;}_{p(X)}(X,N+\num 2)$.
\\[10pt]
\emph{Case~2:} $U>a$.  By~(\ref{mon4a}), the second conjunctive term
of~(\ref{int5}) implies $\st^{X;}_{p(X)}(U,N)$.
Since $p(a)$ and $a<U$, $\st^{X;}_{p(X)}(a,N+\num 1)$ follows by $D_0$.  Then,
since $p(X)$ and $X<a$, $\st^{X;}_{p(X)}(X,N+\num 2)$ follows in a similar way.

\section{Comparison with the original formalization} \label{sec:comparison}

The deductive system from the previous paper on {\sc mgc} programs
\cite{lif22a} operates with formulas over the signature~$\sigma_1$ (that
is,~$\sigma_2$ without $\st$ predicates).
Its definition uses the following notation.
If~$r$ is a precomputed term,~$\bfx$ is a tuple of distinct general
variables, and~$F$ is a formula over~$\sigma_0$, then the expression
$\exists_{\geq r}\bfx F$ stands for
\medskip

\begin{tabular}{ll}
$\exists\bfx_1\cdots\bfx_n\left(\,\bigwedge_{i=1}^n F^\bfx_{\,\bfx_i}
  \land \bigwedge_{i<j}\neg(\bfx_i=\bfx_j)  \right)$
  & if $r=\num n>\num 0$,\\
  $\top$,                   & if $r \leq \num 0$,\\
  $\bot$,                   & if $r>\num n$ for all integers~$n$.
\end{tabular}

\medskip\noindent
Here $\bfx_1,\dots,\bfx_n$ are disjoint tuples of distinct general variables
that do not occur in~$F$.  The symbols~$\top$ and~$\bot$ denote the logical
constants \emph{true}, \emph{false}.  The equality between tuples
$X_1,X_2,\dots$ and $Y_1,Y_2,\dots$ is understood as the conjunction
$X_1=Y_1\land X_2=Y_2\land\cdots$.  The three cases in this definition
cover all precomputed terms~$r$, because the set of numerals is
contiguous (Section~\ref{ssec:programs}).  Similarly,
$\exists_{\leq r}\bfx F$ stands for

\medskip
\begin{tabular}{ll}
$\forall\bfx_1\cdots\bfx_{n+1}\left(\bigwedge_{i=1}^{n+1}
  F^\bfx_{\,\bfx_i}
\to \,\bigvee_{i<j}\bfx_i=\bfx_j\right)$
   & if $r=\num n\geq \num 0$,\\
  $\bot$,                   & if $r < \num 0$,\\
  $\top$,                   & if $r>\num n$ for all integers~$n$.
\end{tabular}
\medskip

By \Defs\ we denote the set of all sentences of the forms
\beq
\forall \bfv \left(\atl^{\bfx;\bfv}_F(\bfv,r)\lrar
  \exists_{\geq r} \bfx F\right)
\eeq{def1}
and
\beq
\forall \bfv \left(\atm^{\bfx;\bfv}_F(\bfv,r)\lrar
  \exists_{\leq r}\bfx F\right).
\eeq{def2}
These formulas are similar to the axioms~$D_1$ of \HTC\
(Section~\ref{ssec:htc}) in the sense that both \Defs\ and~$D_1$
can be viewed as definitions of $\atl$ and $\atm$.  But each formula in \Defs\
refers to a specific value~$r$ of the last argument of $\atl$, $\atm$,
whereas the last argument of $\atl$, $\atm$ in~$D_1$ is a variable.
This difference explains why \HTC\ may be a better tool for
proving strong equivalence than deductive systems with the axioms \Defs.

\begin{theorem}
  \label{thm:HTC=>Defs}
  The formulas \Defs\ are provable in~\HTC.
  \end{theorem}
  The proof of this theorem can be found in Appendix~\ref{sec:proof:thm:thm:HTC=>Defs} (page~\pageref{sec:proof:thm:thm:HTC=>Defs}).
  The formulas \Defs\ are the only axioms of the deductive system from the
  previous publication \cite{lif22a} that are not included in \HTC.
So the theorem above shows that all formulas provable in that system are
provable in \HTC\ as well.

\section{Deductive system $\HTC'$} \label{sec:cl}
 
 In this section, we show that combining~$\tau^*$ with an additional
 syntactic transformation~$\gamma$ allows us to replace \HTC\ by a
 classical first-order theory.
 
The signature~$\sigma'_2$ is obtained from
 the signature~$\sigma_2$ (Section~\ref{ssec:htc}) by adding, for
 every predicate symbol~$p$ other than comparison symbols~(\ref{comp}),
 a new predicate symbol~$p'$ of the same arity.
 The formula
 ${\forall {\bf X}(p({\bf X})\to p'({\bf X}))}$, where $\bf X$ is a tuple
 of distinct general variables, is denoted by $\A(p)$.  The set of
all formulas~$\A(p)$ is denoted by~$\A$.
 
 
 For any formula~$F$ over the signature~$\sigma_2$, by~$F'$ we denote
 the formula over~$\sigma'_2$
 obtained from~$F$ by replacing every occurrence of every
 predicate symbol~$p$ other than comparison symbols by $p'$.  
 The translation~$\gamma$, which relates the logic of here\nobreakdash-and\nobreakdash-there
 to classical logic, maps formulas over~$\sigma_2$ to formulas
 over~$\sigma'_2$.
 It is defined recursively:
 \begin{itemize}
 \item $\gamma F=F$ if~$F$ is atomic,
 \item $\gamma(\neg F)=\neg F'$,
 \item $\gamma(F\land G)=\gamma F \land \gamma G$,
 \item $\gamma(F\lor G)=\gamma F \lor \gamma G$,
 \item $\gamma(F\to G)=(\gamma F \to \gamma G)\land (F'\to G')$,
 \item $\gamma(\forall X\,F)=\forall X\,\gamma F$,
 \item $\gamma(\exists X\,F)=\exists X\,\gamma F$.
 \end{itemize}
To apply~$\gamma$ to a set of formulas means to apply~$\gamma$ to each of
 its members.
 
 By $\HTC'$ we denote the classical first-order theory over the
 signature~$\sigma_2'$ with the axioms~$\A$, $\gamma(\Ind)$,
  \Std, $\gamma D_0$ and $\gamma D_1$.

\begin{theorem} \label{lem:HTC<=>FOC}
   A sentence~$F\lrar G$ over signature~$\sigma_2$ is provable in \HTC\ iff
   $\gamma F \lrar \gamma G$ is provable in $\HTC'$.
\end{theorem}

The proof of this theorem can be found in Appendix~\ref{sec:plem} (page~\pageref{sec:plem}).

\begin{corollary}
   A sentence~$F$ over the signature~$\sigma_2$ is provable in \HTC\ iff
   $\gamma F$ is provable in $\HTC'$.
\end{corollary}

\begin{proof}
  In Theorem~\ref{lem:HTC<=>FOC}, take~$G$ to be~$\top$.
\end{proof}

From Theorems~\ref{thm:j:HTC.strong.equivalence} and~\ref{lem:HTC<=>FOC}
we conclude that {\sc mgc} programs~$\Pi_1$ and~$\Pi_2$ are strongly
equivalent if the formula $\gamma\tau^*\Pi_1\lrar\gamma\tau^*\Pi_2$
is provable in $\HTC'$.

\section{Deductive system $\HTCw$}\label{sec:omega}

In case of the language mini-\gringo, using inference rules with
infinitely many premises allows us to define a deductive system that
satisfies not only condition~(\ref{p}) but also its converse:
programs~$\Pi_1$,~$\Pi_2$ are strongly equivalent \emph{if and only if}
$\tau^*\Pi_1$ and $\tau^*\Pi_2$ can be derived from each other
\cite[Theorem~6]{fan23}.  In this section we define a deductive system
with the same property for the language {\sc mgc}.  This system,
like the deductive system from the previous publication on {\sc mgc}
\cite{lif22a}, does not require extending the signature~$\sigma_1$.

The system $\HTCw$ is an extension of first-order intuitionistic logic
formalized as the natural deduction system \emph{Int}
\cite[Section~5.1]{fan23} for the signature~$\sigma_1$.  Its derivable objects
are \emph{sequents\/}---expressions $\G\seq F$, where $\G$ is a
finite set of formulas over $\sigma_1$ (``assumptions''), and~$F$ is a
formula over~$\sigma_1$.  A sequent of the form $\seq F$ is identified with
the formula~$F$.  The system $\HTCw$ is obtained from \emph{Int}
by adding
\begin{itemize}
\item  axiom schemas (\ref{hosoi}) and (\ref{sqht})
for all formulas $F$, $G$, $H$ over~$\sigma_1$,
\item axioms \Std\  and \Defs, and
  \item the \emph{$\omega$-rules}
    \begin{gather}
      \frac
    {\G\seq F^X_t \hbox{ for all precomputed terms}~t}
    {\G\seq \forall X\,F},
    \label{eq:omega.general}
    \end{gather}
  where~$X$ is a general variable, and
  \begin{gather}
  \frac
  {\G\seq F^N_{\num n} \hbox{ for all integers}~n}
  {\G\seq \forall N\,F},
  \label{eq:omega.integer}
  \end{gather}
where~$N$ is an integer variable.
\end{itemize}

Induction axioms are not on this list, but the instances of the induction
schema \Ind\ for all
formulas~$F$ over~$\sigma_1$ are provable in~$\HTC^\omega$.  Indeed, we
can prove in~$\HTC^\omega$ the sequents
$$
F^N_{\num 0}\land \forall N\left(N\geq\num 0 \land F \to F^N_{N+\num 1}\right)
\seq \num n\geq\num 0 \to F
$$
for all integers~$n$; then \Ind\ can be derived by the
second $\omega$-rule followed by implication introduction.

\begin{theorem}
\label{thm:HTCw.strong.equivalence}
For any {\sc mgc} programs $\Pi_1$ and $\Pi_2$, the
formula $\tau^*\Pi_1\lrar\tau^*\Pi_2$ is provable in $\HTCw$ iff
$\Pi_1$ and $\Pi_2$ are strongly equivalent.
\end{theorem}

The proof of this theorem can be found in Appendix~\ref{sec:proof:thm:HTCw.strong.equivalence} (page~\pageref{sec:proof:thm:HTCw.strong.equivalence}).
The system $\HTCw$ is not an extension of \HTC, because
its axioms say nothing about the predicate symbols $\st^{\bfx;\bfv}_F$.
But all theorems of \HTC\ that do not contain these symbols are
provable in $\HTCw$:

\begin{theorem}
\label{thm:stronger}
Every sentence over the signature~$\sigma_1$ provable in \HTC\ is
provable in $\HTCw$.
\end{theorem}

The proof of this theorem can be found in Appendix~\ref{sec:proof:thm:stronger} (page~\pageref{sec:proof:thm:stronger}).

In order to illustrate how~$\HTCw$ can be used to prove strong equivalence without using the~$\st$ predicates, consider the two programs from the introduction.
As discussed in Section~\ref{sec:example}, 
the strong equivalence of these two programs will be verified if we can prove formula~\eqref{l1}, which is equivalent to the conjunction of formulas
\begin{gather}
\forall N(
  p(a) \wedge 
\atl^{X;}_{p(X)\land X\neq a}(N)\to \atl^{X;}_{p(X)}(N+\num 1)),
\label{l1.l}
\\
\forall N(
  p(a) \wedge 
  \atl^{X;}_{p(X)}(N+\num 1)
  \to 
\atl^{X;}_{p(X)\land X\neq a}(N))
\label{l1.r}
\end{gather}
We can prove~\eqref{l1.l} using the~$\omega$-rule~\eqref{eq:omega.integer} with infinitely many premises of the form
\begin{gather}
  p(a) \wedge 
\atl^{X;}_{p(X)\land X\neq a}(\num n)\to \atl^{X;}_{p(X)}(\num n +\num 1),
\label{l1.l.gr}
\end{gather}
one for each integer~$n$.
Each of these formulas can be proved in~\HTCw, without using the $\omega$-rule.
Note that, by~\eqref{def1}, the antecedent of~\eqref{l1.l.gr} entails
\begin{align*}
p(a) \wedge \exists X_1\cdots X_n\left(\,
  p(X_1) \wedge X_1 \neq a \land \dotsc \land p(X_n) \wedge X_n \neq a
  \land \bigwedge_{i<j}\neg(X_i=X_j)  \right)  
\end{align*}
which entails the consequent of~\eqref{l1.l.gr}.
The proof uses the introduction of the existential quantifier rule, introducing the variable~$X_{n+1}$ in place of the object constant~$a$, and equivalence~\eqref{def1} again, now with~${n+1}$ variables~$X_1,\dotsc,X_{n+1}$.
%
We can similarly prove~\eqref{l1.r} using the
$\omega$-rule~\eqref{eq:omega.integer} with infinitely many premises obtained by replacing in~\eqref{l1.r} variable~$N$ by numeral~$\num n$ for each integer~$n$.

\section{Conclusion}

In this paper, we argue that strong equivalence of two
programs with counting can be established, in many cases, by proving the
equivalence of the corresponding first-order
sentences in the deductive system \HTC.
We do not know whether \HTC\ is
complete for strong equivalence, that is to say, if
\hbox{$\tau^*\Pi_1\lrar\tau^*\Pi_2$} is provable in \HTC\ for all
pairs~$\Pi_1$,~$\Pi_2$ of strongly equivalent {\sc mgc} programs.  But
the deductive system $\HTCw$, which contains infinitary rules, is
complete in this sense.

Sentences~$F_1$,~$F_2$ are equivalent in~\HTC\ if and only if the
sentences $\gamma F_1$,~$\gamma F_2$ are equivalent in the classical
first-order theory~$\HTC'$.  This fact suggests that it may be possible to
use theorem provers for classical theories, such as \vampire\ \cite{vor13},
to verify strong equivalence of {\sc mgc} programs.  Extending the proof
assistant {\sc anthem} \cite{fan20,heu20} in this direction is a
topic for future work.

A translation similar to~$\tau^*$ is used in  {\sc anthem} to
verify another kind of equivalence of mini-\gringo\ programs---equivalence
with respect to a user guide
\cite{fan23a,han23}.  We plan to extend work on user guides
to programs with counting.

Finally, we would like to investigate the possibility of extending the
deductive systems described in this paper to counting aggregates with comparison symbols other that~$\geq$ and~$\leq$, and to aggregates other than counting~\cite{fanhan25a}.

\section*{Acknowledgments}

We would like to thank the anonymous reviewers for their comments that have helped us to improve the paper.
This research is partially supported by NSF CAREER award 2338635.
Any opinions, findings, and conclusions or recommendations expressed in this material are those of the authors and do not necessarily reflect the views of the National Science Foundation.


\bibliographystyle{tlplike}
\bibliography{krr,bib,procs}

\newpage

\appendix

\section{---\hspace{6pt}Review of the~$\tau^*$ translation}
\label{sec:tau*.review}

The target language of the translation~$\tau^*$ is a first-order language with signature~$\sigma_1$ described in Sections~\ref{ssec:programs.as.formulas} and~\ref{ssec:repag}.

\citet{lif19} defined, for every mini-{\sc gringo} term~$t$, a formula $\val tZ$ that expresses, informally speaking, that~$Z$ is one of the values of~$t$.  
\citet{fan23} further refined this definition by introducing the absolute value and providing an alternative definition for the case of division%
\footnote{Grounder \gringo\ \cite{gringomanual} truncates the quotient toward zero, instead of applying the floor function.
This feature of \gringo\ was not taken into account in earlier publications (\citealt[Section~4.2]{geb15};
      \citealt[Section~6]{lif19};
      \citealt[Section~3]{fan20}).}\!\!\~. \
We reproduce here this last definition.

For every mini-{\sc gringo} term $t$, a formula~$\val tZ$ over the signature~$\sigma_0$ is defined, where $Z$ is a general variable that
does not occur in~$t$. 
The definition is recursive:
\begin{itemize}
\item
if $t$ is a precomputed term or a variable then $\val tZ$ is $Z=t$,
\item
  if~$t$ is $|t_1|$ then $\val tZ$ is
  $\exists I(\val{t_1}I\land Z =|I|)$,
\item
if $t$ is $t_1\;\emph{op}\;t_2$, where \emph{op} is $+$, $-$, or $\times$
then $\val tZ$ is
$$\exists I J  (\val{t_1}I \land \val{t_2}J \land Z=I\;\emph{op}\;J),$$
\item
if $t$ is $t_1\,/\,t_2$ then $\val tZ$ is
$$\ba l
\exists I J K (\val{t_1}I \land \val{t_2}J \\
\hskip 1cm \land\; K\times |J|\leq |I|<(K+\num 1)\times |J|\\
\hskip 1cm \land\; ((I\times J \geq \num 0 \land Z=K)\\
\hskip 1cm \lor\;(I\times J < \num 0 \land Z=-K))),
\ea$$
\item
if $t$ is $t_1\backslash t_2$ then $\val tZ$ is
$$\ba l
\exists I J K (\val{t_1}I \land \val{t_2}J \\
\hskip 1cm \land\; K\times |J|\leq |I|<(K+\num 1)\times |J|\\
\hskip 1cm \land\; ((I\times J \geq \num 0 \land Z=I-K\times J)\\
\hskip 1cm \lor\;(I\times J < \num 0 \land Z=I+K\times J))),
\ea$$
\item
if $t$ is $t_1\,..\,t_2$ then $\val tZ$ is
$$\exists I J K (\val{t_1}I \land \val{t_2}J \land
    I\leq K \leq J \land Z=K),$$
    where $I$, $J$, $K$ are fresh integer variables.
\end{itemize}
If $\bf t$ is a tuple $t_1,\dots,t_n$ of mini-{\sc gringo} terms, and
$\bf Z$ is a tuple $Z_1,\dots,Z_n$ of distinct general variables, then
$\val{\bf t}{\bf Z}$ stands for the conjunction
$\val{t_1}{Z_1} \land \cdots \land \val{t_n}{Z_n}$.

The translation~$\tau^B$, which transforms literals
and comparisons into formulas over the signature~$\sigma_0$,
is defined in that paper as
follows:\footnote{The superscript~$B$ indicates that
this translation is intended for \emph{bodies} of rules.}
\begin{itemize}
\item
  $\tau^B(p({\bf t}))=
  \exists {\bf Z}(\val{\bf t}{\bf Z} \land p({\bf Z}))$;
\item
  $\tau^B(\no\ p({\bf t})) =
  \exists {\bf Z}(\val{\bf t}{\bf Z} \land \neg p({\bf Z}))$;
\item
  $\tau^B(\no\ \no\ p({\bf t})) =
  \exists {\bf Z}(\val{\bf t}{\bf Z} \land \neg\neg p({\bf Z}))$;
\item
$\tau^B(t_1\prec t_2) =
\exists Z_1 Z_2 (\val{t_1}{Z_1} \land \val{t_2}{Z_2} \land
Z_1\prec Z_2)$;
\end{itemize}
Here $Z_1$,~$Z_2$, and members of the tuple {\bf Z} are fresh general
variables.

The result of applying $\tau^*$ to a mini-\gringo\ rule
$H \ar B_1\land\cdots\land B_n$
can be defined as the universal closure of the formula
\beq\ba{ll}
B^*_1\land\cdots\land B^*_n\land\val{\bf t}{\bf Z}\to p({\bf Z})
&\hbox{ if }H\hbox{ is }p({\bf t}),\\
B^*_1\land\cdots\land B^*_n\land\val{\bf t}{\bf Z}
  \to p({\bf Z}) \lor \neg p({\bf Z})
&\hbox{ if }H\hbox{ is }p\{({\bf t})\},\\
\neg(B^*_1\land\cdots\land B^*_n)
&\hbox{ if $H$ is empty},
\ea\eeq{bstar}
where~$\bf Z$ is a tuple of
fresh general variables, and~$B^*_i$ stands for $\tau^B(B_i)$ if~$B_i$ does not include and aggregate element, and for
$$
\exists C\left(\val t C
\land \atl^{\bfx;\bfv}_{\exists \bfw\tau^B({\bf L})}(\bfv,C)\right)
$$
if $B_i$ is $\co\{\bfx:{\bf L}\} \geq t$, and as
$$
\exists C\left(\val t C
\land \atm^{\bfx;\bfv}_{\exists \bfw\tau^B({\bf L})}(\bfv,C)\right)
$$
if $B_i$ is $\co\{\bfx:{\bf L}\} \leq t$,
where~$C$ is a fresh general variable.
Here, $\bfv$ is the list of global variables that occur in $\bf L$, and~$\bfw$ is the list of local variables that occur in~$\bf L$ but are not included in~$\bfx$.

\section{---\hspace{6pt}Proofs}
\label{sec:proofs}

In this section, we provide the proofs of the results stated in the main body of the paper.

\subsection{Proofs of Propositions~\ref{prop:simp}--\ref{prop:exactly1}}
\label{sec:proofs-props}


\subsubsection{A few more theorems of \HTC}


The symbols $\leq$ and $<$ between tuples refer to lexicographic order, as in
Section~\ref{ssec:htc}.

\medskip\noindent\emph{Claim:}  If~$\bfx$, $\bfw$ are disjoint tuples of
distinct general variables of the same length, and the variables~$\bfw$
are not free in~$F$, then the formula
\beq
\forall \bfx\bfw\bfv N
(\st^{\bfx;\bfv}_F(\bfx,\bfv,N) \land \bfw\leq \bfx
\land F^\bfx_{\bfw}
\to \st^{\bfx;\bfv}_F(\bfw,\bfv,N))
\eeq{mon1}
is provable in $\HTC$.

\medskip\noindent\emph{Proof.}
By~$D_0$, if $N\leq\num 0$ then
$\st^{\bfx;\bfv}_F(\bfw,\bfv,N)$; also, if $N=\num 1$ then
$$F^\bfx_{\bfw} \to \st^{\bfx;\bfv}_F(\bfw,\bfv,N).$$
It remains to prove
\beq
N> \num 0\land \st^{\bfx;\bfv}_F(\bfx,\bfv,N+\num 1) \land \bfw\leq \bfx
\land F^\bfx_{\bfw}
\to \st^{\bfx;\bfv}_F(\bfw,\bfv,N+\num 1).
\eeq{g0}
(This assertion is justified by the formula
$$\forall N
(N\leq\num 0 \lor N=\num 1 \lor \exists M(N=M+\num 1 \land M>\num 0)),
$$
which belongs to \Std.)
Assume the antecedent of~(\ref{g0}).  From
the first two conjunctive terms, by~$D_0$, we can conclude that
there exists $\bfu$ such that
$$\bfx<\bfu\land \st^{\bfx;\bfv}_F(\bfu,\bfv,N).$$
In combination with the last two conjunctive terms, we get
$$F^\bfx_{\bfw} \land\bfw<\bfu\land \st^{\bfx;\bfv}_F(\bfu,\bfv,N).$$
Now the consequent of~(\ref{g0}) follows by~$D_0$.
\qed

\medskip\noindent\emph{Claim:} The formula
\beq
\forall \bfx\bfv N
(\st^{\bfx;\bfv}_F(\bfx,\bfv,N+\num 1)\to\st^{\bfx;\bfv}_F(\bfx,\bfv,N))
\eeq{mon2}
is provable in $\HTC$. 

\medskip\noindent\emph{Proof.}
  If $N\leq \num 0$ then the consequent of~(\ref{mon2}) follows
  from~$D_0$.
If $N>\num 0$ then, by~$D_0$, the antecedent of~(\ref{mon2})
implies
$$
F\land\exists \bfu(\bfx<\bfu\land \st^{\bfx;\bfv}_F(\bfu,\bfv,N)).
$$
Thus there exists~$\bfu$ such that
$\st^{\bfx;\bfv}_F(\bfu,\bfv,N)\land \bfx<\bfu\land F$.
The consequent of~(\ref{mon2}) follows by~(\ref{mon1}).
\qed

\medskip\noindent\emph{Claim:} The formula
\beq
\forall \bfx\bfv MN
(\st^{\bfx;\bfv}_F(\bfx,\bfv,M)\land M\geq N
\to\st^{\bfx;\bfv}_F(\bfx,\bfv,N))
\eeq{mon3}
is provable in $\HTC$.
 
\medskip\noindent\emph{Proof.}
Since $M\geq N$ is equivalent to $\exists K(K\geq \num 0\land M=N+K)$,
formula~(\ref{mon3}) can be rewritten as
$$\forall K(K\geq \num 0\to\forall \bfx\bfv N
(\st^{\bfx;\bfv}_F(\bfx,\bfv,N+K)
\to\st^{\bfx;\bfv}_F(\bfx,\bfv,N))).
$$
The proof is by induction~\Ind.  The basis
$$\forall \bfx\bfv N
(\st^{\bfx;\bfv}_F(\bfx,\bfv,N+\num 0)
\to\st^{\bfx;\bfv}_F(\bfx,\bfv,N))
$$
follows from the \Std\  axiom $\forall N(N+\num 0=N)$. The induction
hypothesis is
$$
K\geq \num 0\land\forall \bfx\bfv N
(\st^{\bfx;\bfv}_F(\bfx,\bfv,N+K)
\to\st^{\bfx;\bfv}_F(\bfx,\bfv,N));
$$
we need to derive
$$\forall \bfx\bfv N
(\st^{\bfx;\bfv}_F(\bfx,\bfv,N+K+\num 1)
\to\st^{\bfx;\bfv}_F(\bfx,\bfv,N)).
$$
Assume $\st^{\bfx;\bfv}_F(\bfx,\bfv,N+K+\num 1)$.  Then,
by~(\ref{mon2}), $\st^{\bfx;\bfv}_F(\bfx,\bfv,N+K)$, and
$\st^{\bfx;\bfv}_F(\bfx,\bfv,N)$ follows by the induction hypothesis.
\qed


\subsubsection{Proof of Proposition~\ref{prop:simp}}

\noindent\emph{Proof of (\ref{htc:atln}).}
   By~$D_1$, the left-hand side of~(\ref{htc:atln}) is equivalent to
   $$\exists \bfx M(\st^{\bfx;\bfv}_F(\bfx,\bfv,M)\land M\geq N).$$
   From~(\ref{mon3}) we can conclude that
   $$\exists M (\st^{\bfx;\bfv}_F(\bfx,\bfv,M)\land M\geq N)$$
   is equivalent to $\st^{\bfx;\bfv}_F(\bfx,\bfv,N)$.
\qed

\medskip\noindent\emph{Claim:} The formula
\beq
\forall \bfv N(\atm^{\bfx;\bfv}_F(\bfv,N)
\lrar
\neg\exists\bfx\,\st^{\bfx;\bfv}_F(\bfx,\bfv,N+\num 1))
\eeq{atmn}
is provable in $\HTC$.

\medskip\noindent\emph{Proof.}
  By~$D_1$, the left-hand side of~(\ref{atmn}) is equivalent to
  $$\forall \bfx M(\st^{\bfx;\bfv}_F(\bfx,\bfv,M)\to M\leq N)$$
  and consequently to
$$\neg\exists \bfx M(\st^{\bfx;\bfv}_F(\bfx,\bfv,M)\land M>N).$$
The formula
$$\exists M (\st^{\bfx;\bfv}_F(\bfx,\bfv,M)\land M>N)$$
is equivalent to
$$\exists M (\st^{\bfx;\bfv}_F(\bfx,\bfv,M)\land M\geq N+\num 1)$$
and, by~(\ref{mon3}), to
$\st^{\bfx;\bfv}_F(\bfx,\bfv,N+\num 1)$.
\qed

\medskip
Formula~(\ref{htc:atmn1}) follows from~(\ref{htc:atln}) and~(\ref{atmn}).

\subsubsection{Proof of Proposition~\ref{prop:other}}

Formulas~(\ref{atlnonp}) and~(\ref{atl1}) follow from~(\ref{htc:atln})
and~$D_0$.

\medskip\noindent\emph{Claim:} Formula (\ref{mon4a}) is provable in $\HTC$.

\begin{proof}
  Assume $\forall \bfx(F\to G)$.  If $N\leq 0$ then
  $\st^{\bfx;\bfv}_G(\bfx,\bfv,N)$ by~$D_0$.  For positive~$N$,
  the proof is by induction~(\ref{ind1}).
The basis
$$
\st^{\bfx;\bfv}_F(\bfx,\bfv,\num 1) \to \st^{\bfx;\bfv}_G(\bfx,\bfv,\num 1)
$$
is equivalent to $F\to G$ by~$D_0$.  Take a positive~$N$ and assume
$$\forall\bfx\bfv
(\st^{\bfx;\bfv}_F(\bfx,\bfv,N) \to \st^{\bfx;\bfv}_G(\bfx,\bfv,N));
$$
we want to show that
\beq
\forall\bfx\bfv
(\st^{\bfx;\bfv}_F(\bfx,\bfv,N+\num 1) \to
\st^{\bfx;\bfv}_G(\bfx,\bfv,N+\num 1)).
\eeq{g5}
Assume $\st^{\bfx;\bfv}_F(\bfx,\bfv,N+\num 1)$.  By~$D_0$,
$$F\land\exists \bfu(\bfx<\bfu\land \st^{\bfx;\bfv}_F(\bfu,\bfv,N)).$$
Then~$G$ and, by the induction hypothesis,
$$\exists \bfu(\bfx<\bfu\land \st^{\bfx;\bfv}_G(\bfu,\bfv,N)).$$
The consequent of~(\ref{g5}) follows by~$D_0$.
\end{proof}

\medskip\noindent\emph{Claim:} Formula (\ref{stpos}) is provable in $\HTC$.

\begin{proof}
Since
$$
\forall N( N>\num 0 \to N=\num 1 \lor \exists M(N=M+\num 1 \land M>\num 0))
$$
(Group \Std\ axiom), it is sufficient to show that
$$
\forall \bfz(\st^{\bfx;\bfv}_F(\bfz,\bfv,\num 1) \to F^\bfx_\bfz)
$$
and
$$\forall \bfz\bfv N (\st^{\bfx;\bfv}_F(\bfz,\bfv,N+\num 1)\land N>\num 0
\to F^\bfx_\bfz).
$$
Both formulas follow from axioms~$D_0$.
\end{proof}

\subsubsection{Proof of Proposition~\ref{prop:exactly1}}

\noindent\emph{Proof of (\ref{exan}).}
Assume $\exa^{\bfx;\bfv}_F(\bfx,Y)$.
Then $\atm^{\bfx;\bfv}_F(\bfx,Y)$.
From~$D_0$ and~$D_1$ we can conclude that
$$\forall \bfv Y(\atm^{\bfx;\bfv}_F(\bfv,Y)\to \num 0\leq Y).$$
Hence $\num 0 \leq Y$.  On the other hand, $\atl^{\bfx;\bfv}_F(\bfx,Y)$,
so that $\exists N(N\geq Y)$ by~$D_1$.
Thus $\exists N(\num 0\leq Y\leq N)$.  It remains to observe that the formula
$$
\forall Y(\exists N(\num 0\leq Y\leq N)\leq \exists N(Y=N\land N\geq\num 0))
$$
is a group \Std\ axiom, because the set of numerals is contiguous.
\qed

\medskip\noindent\emph{Claim:} The formula
\beq
\forall \bfx(F\to G)\to
\forall\bfx Y(\atl^{\bfx;\bfv}_F(\bfx,Y)\to\atl^{\bfx;\bfv}_G(\bfx,Y))
\eeq{mon4b}
is provable in $\HTC$.

\begin{proof}
  Assume $\forall \bfx(F\to G)$ and $\atl^{\bfx;\bfv}_F(\bfx,Y)$.
  By~$D_1$,
  $$\exists \bfx N(\st^{\bfx;\bfv}_F(\bfx,\bfv,N)\land N\geq Y).$$
  Then
  ${\exists \bfx N(\st^{\bfx;\bfv}_G(\bfx,\bfv,N)\land N\geq Y)}$
  by~(\ref{mon4a}),
  and $\atl^{\bfx;\bfv}_G(\bfx,Y)$ follows by~$D_1$.%
\end{proof}

\medskip\noindent\emph{Claim:} The formula
\beq
\forall \bfx(F\to G)\to
\forall\bfx Y(\atm^{\bfx;\bfv}_G(\bfx,Y)\to\atm^{\bfx;\bfv}_F(\bfx,Y))
\eeq{mon4c}
is provable in $\HTC$.

\begin{proof}
  Assume $\forall \bfx(F\to G)$ and $\atm^{\bfx;\bfv}_G(\bfx,Y)$.
  By~$D_1$,
  $$\forall \bfx N(\st^{\bfx;\bfv}_G(\bfx,\bfv,N)\to N\leq Y)).$$
  Then
  $\forall \bfx N(\st^{\bfx;\bfv}_F(\bfx,\bfv,N)\to N\leq Y))$
  by~(\ref{mon4a}),
  and $\atm^{\bfx;\bfv}_F(\bfx,Y)$ follows by~$D_1$.
\end{proof}

Formula~(\ref{mon4d}) follows from~(\ref{mon4b}) and~(\ref{mon4c}).

\subsection{Proof of Theorem~\ref{thm:HTC=>Defs}}
\label{sec:proof:thm:thm:HTC=>Defs}

\subsubsection{A few more theorems of \HTC, continued}
\label{ssec:continued}

\medskip\noindent\emph{Claim:} Let~$F$ be a formula over~$\sigma_0$,
let~$\bfu$, $\bfw$ are disjoint tuples of
distinct general variables of the same length such that  the variables~$\bfw$
are not free in~$F$, and let~$n$ is a positive integer.  The formula
\beq
\exists_{\geq \num{n+1}}\bfu\,F \lrar
\exists\bfu(F\land\exists_{\geq \num n}\bfw(\bfu<\bfw \land F^\bfu_\bfw))
\eeq{htap:explus1}
is provable in~$\HTC$.

\begin{proof} \emph{Left-to-right:} take $\bfu_1,\dots,\bfu_{n+1}$ such that
\beq
\bigwedge_{i=1}^{n+1} F^\bfu_{\,\bfu_i}
  \land \bigwedge_{i<j}\neg(\bfu_i=\bfu_j).
\eeq{a1}
We reason by cases, using the  axiom
$$\bigvee_{k=1}^{n+1}\bigwedge_{i=1}^{n+1}\bfu_k\leq\bfu_i$$
from \Std\ (``for some~$k$,
$\bfu_k$ is lexicographically first among $\bfu_1,\dots,\bfu_{n+1}$'').
Consider the $k$-th case
$\bigwedge_{i=1}^{n+1}\bfu_k\leq\bfu_i$.  From~(\ref{a1}),
$$\bfu_k<\bfu_i\hbox{ and }F^\bfu_{\bfu_i} \qquad(i=1,\dots,n+1;\ i\neq k)$$
and
$$\neg(\bfu_i=\bfu_j)\qquad (1\leq i<j\leq n+1;\ i,j\neq k).$$
Hence
$\exists_{\geq\num n} \bfw(\bfu_k<\bfw\land F^\bfu_{\,\bfw})$.
Since $F^\bfu_{\,\bfu_k}$, it follows that
$$\exists\bfu(F\land\exists_{\geq \num n}\bfw(\bfu<\bfw \land F^\bfu_\bfw)).$$
\emph{Right-to-left:} assume
$$
\exists\bfu\left(F\land\exists\bfw_1\cdots\bfw_n\left(\,\bigwedge_{i=1}^n
  (\bfu<\bfw_i\land F^\bfu_{\,\bfw_i}) \land \bigwedge_{i<j}\neg(\bfw_i=\bfw_j) \right)\right).
$$
This formula is equivalent to
$$
\exists\bfu\bfw_1\cdots\bfw_n\left(F\land\bigwedge_{i=1}^n
  (\bfu<\bfw_i\land F^\bfu_{\,\bfw_i}) \land \bigwedge_{i<j}\neg(\bfw_i=\bfw_j) \right)
$$
and can be rewritten as
$$
\exists\bfw_0\bfw_1\cdots\bfw_n\left(F^\bfu_{\,\bfw_0}\land\bigwedge_{i=1}^n
  (\bfw_0<\bfw_i\land F^\bfu_{\,\bfw_i}) \land \bigwedge_{1\leq i<j\leq n}\neg(\bfw_i=\bfw_j) \right).
$$
It implies
$$
\exists\bfw_0\bfw_1\cdots\bfw_n\left(\bigwedge_{i=0}^n
  F^\bfu_{\,\bfw_i} \land \bigwedge_{0\leq i<j\leq n}\neg(\bfw_i=\bfw_j) \right),
$$
which is equivalent to $\exists_{\geq \num{n+1}}\bfu\,F$.
\end{proof}

\medskip\noindent\noindent\emph{Claim:}
  If~$\bfx$, $\bfu$ are disjoint tuples of
distinct general variables of the same length, the variables~$\bfu$
are not free in~$F$, and~$n>0$, then the sentence
\beq
\forall\bfx\bfv(
\exists \bfu(\bfx<\bfu\land \st^{\bfx;\bfv}_F(\bfu,\bfv,\num n))
\lrar
\exists_{\geq \num n}\bfu(\bfx<\bfu \land F^\bfx_\bfu))
\eeq{htcn:tri}
is provable in~$\HTC$.

\begin{proof}
  By induction on~$n$.  If~$n=1$ then
  $\st^{\bfx;\bfv}_F(\bfu,\bfv,\num n)$ in the left-hand side
  of~(\ref{htcn:tri}) is equivalent to $F^\bfx_\bfu$
  by~$D_0$, and
  the right-hand side of~(\ref{htcn:tri}) is equivalent to
$\exists \bfu(\bfx<\bfu\land F^\bfx_\bfu)$.  Induction step: we will show that
the formula
\beq
\forall\bfx\bfv
(\exists \bfu(\bfx<\bfu\land \st^{\bfx;\bfv}_F(\bfu,\bfv,\num {n+1}))
\lrar
\exists_{\geq \num{n+1}}\bfu(\bfx<\bfu \land F^\bfx_\bfu))
\eeq{g4}
is derivable from~(\ref{htcn:tri}) in~$\HTC$.
By~(\ref{htap:explus1}), the right-hand side of~(\ref{g4}) is equivalent to
\beq
\exists\bfu(\bfx<\bfu \land F^\bfx_\bfu\land
\exists_{\geq \num n}\bfw(\bfu<\bfw \land \bfx<\bfw\land F^\bfx_\bfw)).
\eeq{s1}
In the presence of $\bfx<\bfu$, the subformula $\bfu<\bfw \land \bfx<\bfw$
is equivalent to $\bfu<\bfw$.  Hence~(\ref{s1}) is equivalent to
\beq
\exists\bfu(\bfx<\bfu \land F^\bfx_\bfu\land
\exists_{\geq \num n}\bfw(\bfu<\bfw \land F^\bfx_\bfw)).
\eeq{s2}
On the other hand,~(\ref{htcn:tri}) can be rewritten as
$$
\forall\bfx\bfv(
\exists \bfw(\bfx<\bfw\land \st^{\bfx;\bfv}_F(\bfw,\bfv,\num n))
\lrar
\exists_{\geq \num n}\bfw(\bfx<\bfw \land F^\bfx_\bfw)),
$$
and it implies
$$
\exists \bfw(\bfu<\bfw\land \st^{\bfx;\bfv}_F(\bfw,\bfv,\num n))
\lrar
\exists_{\geq \num n}\bfw(\bfu<\bfw \land F^\bfx_\bfw).
$$
It follows that~(\ref{s2}) is equivalent to
$$
\exists\bfu(\bfx<\bfu \land F^\bfx_\bfu\land
\exists \bfw(\bfu<\bfw\land \st^{\bfx;\bfv}_F(\bfw,\bfv,\num n))).
$$
By~$D_0$, this formula is equivalent to the left-hand side
of~(\ref{g4}).
\end{proof}

\medskip\noindent\noindent\emph{Claim:}
If~$\bfu$ is a tuple of distinct general variables of the same length
as~$\bfx$ such that its members do not belong to~$\bfx$ and are not free
in~$F$ then the sentence
\begin{gather}
\forall\bfv(  \exists \bfu 
    \st^{\bfx,\bfv}_F(\bfu,\bfv, \num n)
    \lrar
    \exists_{\geq \num n} \bfu \, F^\bfx_\bfu)
  \label{htcn:st<->defs}
\end{gather}
is provable in~$\HTC$.

\begin{proof}
The following is one of the axioms of~\Std:
$$
  \num n \leq \num 0 \,\vee\, \num n = \num1 \,\vee\,  \num n > \num 1.
$$
\emph{Case~1:} $\num n \leq \num 0$.
The right-hand side of~\eqref{htcn:st<->defs} is~$\top$ and its left-hand side follows from~$D_0$.
\emph{Case~2:} $\num n = \num 1$.
The right-hand side of~\eqref{htcn:st<->defs} is~$\exists \bfx_1 F^\bfx_{\bfx_1}$ and~\eqref{htcn:st<->defs}
is immediate from~$D_0$.
\emph{Case~3:} $\num n> \num 1$.
The left-hand side of~\eqref{htcn:st<->defs} is equivalent to
$$
\exists\bfx(F\land
\exists \bfu(\bfx<\bfu\land \st^{\bfx;\bfv}_F(\bfu,\bfv,\num{n-1})))
$$
by~$D_0$, and the right-hand side is equivalent to
$$\exists\bfx(F\land\exists_{\geq \num{n-1}}\bfu(\bfx<\bfu \land F^\bfx_\bfu))$$
by~(\ref{htap:explus1}).  These two formulas are equivalent to each other
by~(\ref{htcn:tri}).
\end{proof}

\noindent\emph{Claim:} For any formula~$F$ over~$\sigma_0$ and any
nonnegative integer~$n$, the formula
\beq
\exists_{\leq\num n}\bfx\,F\lrar \neg\exists_{\geq\num{n+1}}\bfx\,F.
\eeq{htap:exex}
is provable in~$\HTC$.

\begin{proof}
$$\ba{rcl}
\neg\exists_{\geq\num{n+1}}\bfx\,F
&=&\neg\exists\bfx_1\cdots\bfx_{n+1}\left(\bigwedge_{i=1}^{n+1}
  F^\bfx_{\,\bfx_i}\land \,\bigwedge_{i<j}\neg(\bfx_i=\bfx_j)\right)\\
&\lrar&\forall\bfx_1\cdots\bfx_{n+1}\neg\left(\bigwedge_{i=1}^{n+1}
  F^\bfx_{\,\bfx_i}\land \,\bigwedge_{i<j}\neg(\bfx_i=\bfx_j)\right)\\
&\lrar&\forall\bfx_1\cdots\bfx_{n+1}\left(\bigwedge_{i=1}^{n+1}
  F^\bfx_{\,\bfx_i}\to\neg\,\bigwedge_{i<j}\neg(\bfx_i=\bfx_j)\right)\\
&\lrar&\forall\bfx_1\cdots\bfx_{n+1}\left(\bigwedge_{i=1}^{n+1}
  F^\bfx_{\,\bfx_i}\to \,\bigvee_{i<j}\bfx_i=\bfx_j\right)\\
&=&\exists_{\leq\num n}\bfx\,F.
\ea$$
\end{proof}

\subsubsection{Proof of Theorem~\ref{thm:HTC=>Defs}, Part 1}

We will show now that for every precomputed term~$r$,
sentence~(\ref{def1}) is provable in \HTC.

\medskip\noindent\emph{Case~1:} $r\leq\num 0$; (\ref{def1}) is
  \beq
  \forall \bfv \left(\atl^{\bfx;\bfv}_F(\bfv,r)\lrar \top\right).
  \eeq{g1}
  From~$D_1$,
  $$ \forall \bfv \bfx
  (\st^{\bfx;\bfv}_F(\bfx,\bfv,\num 0)
  \land \num 0\geq r \to \atl^{\bfx;\bfv}_F(\bfv,r)).
  $$
  The conjunctive term $\st^{\bfx;\bfv}_F(\bfx,\bfv,\num 0)$ follows
  from~$D_0$, and second conjunctive term $\num 0\geq r$
  is an axiom of~\Std.
  Consequently $\atl^{\bfx;\bfv}_F(\bfv,r)$, which is equivalent to~(\ref{g1}).
  
  \medskip\noindent\emph{Case~2:} for all~$n$, $r>\num n$; (\ref{def1}) is
  \beq
  \forall \bfv \left(\atl^{\bfx;\bfv}_F(\bfv,r)\lrar \bot\right).
  \eeq{g2}
  Assume $\atl^{\bfx;\bfv}_F(\bfv,r)$. From~$D_1$,
  $$
  \forall \bfv (\atl^{\bfx;\bfv}_F(\bfv,r)
  \to
  \exists \bfx N(\st^{\bfx;\bfv}_F(\bfx,\bfv,N)\land N\geq r)).
  $$
  Consequently $\exists N(N\geq r)$, which contradicts the~\Std\ axiom
  $\forall N\neg(N\geq r)$.
  
  \medskip\noindent\emph{Case~3:} for some~$n$, $\num 0 < r \leq\num n$.
  Since the set of numerals is contiguous,~$r$ is a
  numeral~$\num m$ ($m>0$).  By~(\ref{htc:atln}),  formula (\ref{def1})
  can be rewritten as
$$
  \forall \bfv \left(\exists\bfx\,\st^{\bfx;\bfv}_F(\bfx,\bfv,\num m)\lrar
    \exists_{\geq \num m} \bfx F\right),
$$
  which is the universal closure of~\eqref{htcn:st<->defs}.

\subsubsection{Proof of Theorem~\ref{thm:HTC=>Defs}, Part 2}

We will show now that for every precomputed term~$r$,
sentence~(\ref{def2}) is provable in \HTC.
  
  \medskip\noindent\emph{Case~1:} $r<\num 0$; (\ref{def2}) is
  \beq
  \forall \bfv \left(\atm^{\bfx;\bfv}_F(\bfv,r)\lrar \bot\right).
  \eeq{g1a}
  From~$D_1$,
  $$
  \atm^{\bfx;\bfv}_F(\bfv,r)
  \to
  \forall \bfx (\st^{\bfx;\bfv}_F(\bfx,\bfv,\num 0)\to \num 0\leq r).
  $$
  By~$D_0$, $\st^{\bfx;\bfv}_F(\bfx,\bfv,\num 0)$, so that
  $$
  \atm^{\bfx;\bfv}_F(\bfv,r)
  \to \num 0\leq r.
  $$
  From the \Std\ axiom $\neg(\num 0\leq r)$ we conclude that
  $\neg\atm^{\bfx;\bfv}_F(\bfv,r)$, which is equivalent to~(\ref{g1a}).
  
  \medskip\noindent\emph{Case~2:} for all~$n$, $r>\num n$; (\ref{def2}) is
  \beq
  \forall \bfv \left(\atm^{\bfx;\bfv}_F(\bfv,r)\lrar \top\right).
  \eeq{g2a}
  From~$D_1$,
  $$
  \forall \bfv (\forall \bfx N(\st^{\bfx;\bfv}_F(\bfx,\bfv,N)\to N\leq r)
  \to
  \atm^{\bfx;\bfv}_F(\bfv,r)).
  $$
  The antecedent
  $\forall \bfx N(\st^{\bfx;\bfv}_F(\bfx,\bfv,N)\to N\leq r)$
  follows from the \Std\ axiom $\forall N(N\leq r)$.
  Hence $\atm^{\bfx;\bfv}_F(\bfv,r)$, which is equivalent
  to~(\ref{g2a}).
  
  \medskip\noindent\emph{Case~3:} for some~$n$, $\num 0 \leq r \leq\num n$.
  Since the set of numerals is contiguous~$r$ is a
  numeral~$\num m$ ($m\geq 0$), so that~(\ref{def2}) is
  $$\forall \bfv \left(\atm^{\bfx;\bfv}_F(\bfv,\num m)\lrar
    \exists_{\leq \num m}\bfx F\right).$$
  By (\ref{htc:atmn1}) and (\ref{htap:exex}), this formula is equivalent to
  $$\forall \bfv \left(\neg \atl^{\bfx;\bfv}_F(\bfv,\num {m+1})\lrar
    \neg\exists_{\geq \num {m+1}}\bfx F\right),$$
  which follows from (\ref{def1}).

\subsection{Review: HT\nobreakdash-interpretations}\label{sec:review}

A \emph{propositional HT-interpretation\/}
is a pair
$\tuple{\X,\Y}$, where~$\Y$ is a set of propositional atoms, and~$\X$
is a subset
of~$\Y$.  In terms of Kripke models with two worlds,~$\X$ is
the here-world, and~$\Y$ is the there-world.
The recursive definition of the satisfaction
relation between HT\nobreakdash-interpretations and propositional formulas can be
extended to infinitary propositional formulas \cite[Definition~2]{tru12}.
Equilibrium models of a set of formulas \cite{pea97,pea99} are defined as
its HT-models satisfying a certain minimality
condition.  A set~$\X$ of atoms is a stable model of a set of
infinitary propositional formulas iff
$\tuple{\X,\X}$ is an equlibrium model of that set~\cite[Theorem~3]{tru12}.
Thus stable models of
an {\sc mgc} program~$\Pi$ can be characterized as sets~$\X$ such that
$\tuple{\X,\X}$ is an equilibrium model of~$\tau\Pi$.

The definition of a many-sorted \HTi\ \cite[Appendices~A and~B]{fan22}
extends this construction to many-sorted first-order languages.
In classical semantics of first-order formulas,
the recursive definition of the satisfaction
relation between an interpretation~$I$ of a signature~$\sigma$ and a
sentence~$F$ over~$\sigma$ involves
extending~$\sigma$ by new object constants~$d^*$, which represent elements~$d$
of the domain of~$I$
\cite[Section~1.2.2]{lif08b}.  The extended signature is denoted by~$\sigma^I$.
In the definition of a many-sorted \HTi, the predicate symbols
of~$\sigma$ are assumed to be  partitioned into extensional and intensional.
For any interpretation~$I$ of such a signature~$\sigma$,
$I^\downarrow$ stands for the set of all atomic sentences over~$\sigma^I$
that have the form $p(\boldd^*)$, where~$p$ is intensional, $\boldd$ is
a tuple of elements of appropriate domains of~$\sigma$, and
$I\models p(\boldd^*)$.
An \emph{\HTi}  of a many-sorted signature~$\sigma$
is a pair $\tuple{\HH,I}$, where~$I$ is an
interpretation of~$\sigma$, and $\HH$ is a subset of~$I^\downarrow$.  In
terms of Kripke models,~$I$ is the there-world, and~$\HH$ describes
the intensional predicates in the here-world.

The satisfaction
relation between \HTis\ and sentences is denoted by~$\modelsht$, to
distinguish it from classical satisfaction. According to the persistence
property of this relation, $\tuple{\HH,I}\modelsht F$ implies
$I\models F$ for every sentence~$F$ over~$\sigma$
\cite[Proposition~3a]{fan22}.

The soundness and completeness theorem for the many-sorted logic of
here\nobreakdash-and\nobreakdash-there \cite[Theorem~2]{fan23} can be stated as follows:

\medskip\noindent\emph{For
  any set~$\Gamma$ of sentences over a many-sorted signature~$\sigma$ and any
sentence~$F$ over~$\sigma$, the following two conditions are equivalent:
\begin{itemize}
\item[(i)]
  every \HTi\ of~$\sigma$ satisfying~$\Gamma$ satisfies~$F$;
\item[(ii)]
  $F$ can be derived from~$\Gamma$ in first-order intuitionistic logic
  extended by
\begin{itemize}
\item
  axiom schemas~(\ref{hosoi}) and~(\ref{sqht}) for all formulas~$F$ and~$G$
  over~$\sigma$;
\item
  the axioms
  \beq
  X=Y\lor X\neq Y,
  \eeq{em1}
  where~$X$ and~$Y$ are variables of the same
  sort;
  \item
    the axioms
    \beq
    p(\bfx)\lor \neg p(\bfx),
    \eeq{em2}
    where~$p$ is an extensional
    predicate symbol, and~$\bfx$ is a tuple of distinct variables of
    appropriate sorts.
  \end{itemize}
\end{itemize}}
\noindent
The deductive system described in clause~(ii) is denoted by $\SQHT^=$
\cite[Section~5.1]{fan23}.

\subsection{Proof of Theorem~\ref{lem:HTC<=>FOC}} \label{sec:plem}

In the special case of the
signature~$\sigma_2$ (Section~\ref{ssec:htc}), we designate
comparison symbols~(\ref{comp}) as extensional, and all
other predicate symbols (that is, $p/n$, $\atl^{\bfx;\bfv}_F$,
$\atm^{\bfx;\bfv}_F$ and $\st^{\bfx;\bfv}_F$) as intensional.  This
convention allows us to generalize some of the definitions from
Section~\ref{sec:cl} to arbitrary many-sorted signatures
with predicate constants classified into extensional and intensional.
For any such signature~$\sigma$, by~$\sigma'$ we denote the signature
obtained from it by adding, for every intensional constant~$p$, a new
predicate constant~$p'$ of the same arity.  The formula
\beq
\forall {\bf X}(p({\bf X})\to p'({\bf X})),
\eeq{afla}
where~$p$ is intensional and $\bf X$ is a tuple
of distinct variables of appropriate sorts, is denoted by $\A(p)$, and~$\A$
stands for the set of these formulas for all intensional predicate
constants~$p$.
 For any formula~$F$ over the signature~$\sigma$, by~$F'$ we denote
 the formula over~$\sigma'$
 obtained from~$F$ by replacing every occurrence of every intensional
 predicate symbol~$p$ by $p'$.  
 Then the transformation~$\gamma$ is defined as in Section~\ref{sec:cl}.

For any \HTi\ $\HI$ of~$\sigma$,
$I^\HH$ stands for the interpretation of~$\sigma'$ that has the same domains
as~$I$, interprets function constants and extensional predicate
constants of~$\sigma$ in the same way as~$I$, and interprets the other
predicate constants~$p$,~$p'$ as follows:
\beq
\ba l
I^\HH\models p({\bf d}^*)\hbox{ iff }p({\bf d}^*)\in\HH;\\
I^\HH\models p'({\bf d}^*)\hbox{ iff }I\models p({\bf d}^*).
\ea
\eeq{defexp}
From the second line of~(\ref{defexp}) we can derive a more general
assertion:
\beq
I^\HH\models p'(\boldt)\hbox{ iff }I\models p(\boldt)
\eeq{defexp2}
for every tuple $\boldt$ of ground terms over the signature~$\sigma^I$.
Indeed, the value assigned to~$\boldt$ by the
interpretation~$I^\HH$ (symbolically,~$\boldt^{I^\HH}$)
is the same as the value~$\boldt^I$, assigned to~$\boldt$ by~$I$,
because~$I^\HH$ and~$I$ interpret all symbols
occurring in $\boldt$ in the same way.
In the second line of~(\ref{defexp}), take~$\boldd$ to be the common value of
$\boldt^{I^\HH}$ and $\boldt^I$.  Then
$$I^\HH\models p'\left(\left(\boldt^{I^\HH}\right)^*\right)
\hbox{ iff }I\models p\left(\left({\boldt^I}\right)^*\right),$$
which is equivalent to~(\ref{defexp2}).

\begin{lemma}\label{lem:11}
  An interpretation of the signature~$\sigma'$ satisfies~$\A$ iff it can be
  represented in the form $I^\HH$ for some \HTi~$\HI$.
\end{lemma}

\begin{proof}
  For the if-part, take any formula~(\ref{afla}) from~$\A$.  We need to show
  that $I^\HH$ satisfies all sentences of the form
  $p(\boldd^*)\to p'(\boldd^*)$.  Assume that $I^\HH\models p(\boldd^*)$.
Then $p(\boldd^*)\in\HH\subseteq I^\downarrow$,
  and consequently $I\models p(\boldd^*)$, which is equivalent to
  $I^\HH\models p'(\boldd^*)$.

  For the only-if part, take any
  interpretation~$J$ of~$\sigma'$ that satisfies~$\A$.  Let~$I$ be the
  interpretation of~$\sigma$ that has the same domains as~$J$,
  interprets function constants and extensional predicate constants
  in the same way as~$J$, and interprets every intensional~$p$ in
  accordance with the condition
  \beq
  I\models p({\bf d}^*)\hbox{ iff }J\models p'({\bf d}^*).
  \eeq{iff}
  Take~$\HH$ to be the set of all atoms of the form $p(\boldd^*)$ with
  intensional~$p$ that are satisfied by~$J$.  Since~$J$ satisfies~$\A$,~$J$
  satisfies $p'(\boldd^*)$ for every atom~$p(\boldd^*)$ from~$\HH$.
  By~(\ref{iff}), it follows that all atoms from~$\HH$ are satisfied
  by~$I$, so that~$\HH$ is a
  subset of~$I^\downarrow$.  It follows that $\HI$ is an HT-interpretation.
  Let us show that $I^\HH=J$.  Each of the interpretations~$I^\HH$ and~$J$
  has the same domains as~$I$ and interprets all function constants and
  extensional predicate constants in the same way as~$I$.  For every
  intensional~$p$ and any tuple~$\boldd$ of elements of appropriate
  domains, each of the conditions $I^\HH\models p(\boldd^*)$,
  $J\models p(\boldd^*)$ is equivalent to $p(\boldd^*)\in\HH$,
    and each of the conditions $I^\HH\models p'(\boldd^*)$,
  $J\models p'(\boldd^*)$ is equivalent to $I\models p(\boldd^*)$.
\end{proof}

\begin{lemma}\label{lem:prima.transformation}
  For every \HTi~$\HI$ of~$\sigma$ and every sentence~$F$ over the
  signature~$\sigma^I$,
  $I^\HH\models F'$ iff $I\models F$.
\end{lemma}

\begin{proof}
  We will consider the case when~$F$ is a ground
  atom~$p(\boldt)$; extension to arbitrary sentences by induction is
  straightforward.
  If~$p$ is intensional then~$F'$ is $p'(\boldt)$, so that
  the assertion of the lemma turns into property~(\ref{defexp2}).
  If $p$ is extensional then~$F'$ is $p(\boldt)$; 
  $I^\HH\models p(\boldt)$ iff $I\models p(\boldt)$ because $I^\HH$
  interprets all symbols occurring in~$F$ in the same way as~$I$.
\end{proof}

\begin{lemma}\label{lem:prima.transformation.ht}
  For every \HTi\ $\HI$
  of~$\sigma$ and every sentence~$F$ over the signature~$\sigma^I$,
  $I^\HH\models \gamma F$ iff  $\HI\modelsht F$.
\end{lemma}

\begin{proof}
  The proof is by induction on the number of propositional connectives and
  quantifiers in~$F$.  We consider below the more difficult cases when~$F$ is
an atomic formula, a negation, or an implication.

\medskip\noindent\emph{Case~1:}
$F$ is an atomic formula~$p(\boldt)$.  Then~$\gamma F$
is~$p(\boldt)$ too; we need to check that
\beq
I^\HH\models p(\boldt)\hbox{ iff }\HI\modelsht p(\boldt).
\eeq{goal}
\emph{Case~1.1:}~$p$ is intensional.
Let~$\boldd$ be the common value of $\boldt^{I^\HH}$ and $\boldt^I$.
The left-hand side of~(\ref{goal}) is equivalent to
$I^\HH\models p(\boldd^*)$ and consequently to $p(\boldd^*)\in\HH$.
The right-hand side of~(\ref{goal}) is equivalent to
$p\left(\left(\boldt^I\right)^*\right)\in \HH$, which is
equivalent to $p(\boldd^*)\in\HH$ as well.

\medskip\noindent\emph{Case~1.2:}~$p$ is extensional.
Each of the conditions
$I^\HH\models p(\boldt)$, $\HI\modelsht p(\boldt)$ is equivalent
to~$I\models p(\boldt)$.

\medskip\noindent\emph{Case~2:} $F$ is $\neg G$.  Then $\gamma F$ is
$\neg G'$; we need to check that
$$I^\HH\not\models G\hbox{ iff }\HI\modelsht\neg G'.$$
By Lemma~\ref{lem:prima.transformation}, the left-hand side is
equivalent to $I\not\models G'$. By the definition of~$\modelsht$,
the right-hand side is equivalent to $I\not\models G'$ as well.

\medskip\noindent\emph{Case~3:} $F$ is $G \to H$.  Then $\gamma F$
is $(\gamma G \to \gamma H) \land (G' \to H')$, so that the condition
$I^\HH\models \gamma F$ holds iff
\beq
I^\HH\not\models \gamma G\hbox{ or }I^\HH\models \gamma H
\eeq{d1}
and
\beq
I^\HH\models G'\to H'.
\eeq{d2}
By the induction hypothesis,~(\ref{d1}) is equivalent to
\beq
\HI\not\modelsht G\hbox{ or }\HI\modelsht H.
\eeq{d3}
By Lemma~\ref{lem:prima.transformation},~(\ref{d2}) is equivalent to
\beq
I\models G\to H.
\eeq{d4}
The conjunction of~(\ref{d3}) and~(\ref{d4}) is equivalent to
$\HI\modelsht G\to H$.
\end{proof}

\noindent\emph{Proof of Theorem~\ref{lem:HTC<=>FOC}}.
To prove a formula in \HTC\ means to derive it in first-order
intuitionistic logic from~(\ref{hosoi}),~(\ref{sqht}), \Std, \Ind, $D_0$ and
$D_1$.
   Since the universal closures of~(\ref{em1}) and~(\ref{em2}) belong to
 \Std, it follows that a formula is provable in \HTC\ iff
 it can be derived from \Std, \Ind, $D_0$ and~$D_1$ in $\SQHT^=$.
 Consequently $F\lrar G$ is provable in~\HTC\ iff
 \beq
 \hbox{$G$ can be derived in $\SQHT^=$ from \Std, \Ind, $D_0$, $D_1$ and~$F$}
 \eeq{plr}
and
 \beq
 \hbox{$F$ can be derived in $\SQHT^=$ from \Std, \Ind, $D_0$, $D_1$ and~$G$}.
 \eeq{prl}
 By the soundness and completeness theorem quoted in Appendix~\ref{sec:review},
(\ref{plr}) is equivalent to the condition
$$\ba c
\hbox{$G$ is satisfied by every \HTi\ of~$\sigma_2$}\\
\hbox{that satisfies \Std, \Ind, $D_0$, $D_1$ and~$F$.}
\ea$$
By Lemma~\ref{lem:prima.transformation.ht}, this condition can be further
reformulated as follows:
$$\ba c
\hbox{for every HT\nobreakdash-interpretation~$\HI$ of~$\sigma_2$,
  $I^\HH$ satisfies $\gamma G$}\\
\hbox{if~$I^\HH$ satisfies~$\gamma(\Ind)$, $\gamma(\Std)$,
$\gamma D_0$,~$\gamma D_1$ and~$\gamma F$.}
\ea$$
Then, by Lemma~\ref{lem:11}, (\ref{plr}) is equivalent to the condition
$$\ba c
    \hbox{$\gamma G$ is satisfied by every interpretation of~$\sigma_2'$}\\
    \hbox{that satisfies $\A$, $\gamma(\Ind)$, $\gamma(\Std)$, $\gamma D_0$,
 $\gamma D_1$ and $\gamma F$}.
\ea$$
Similarly, (\ref{prl}) is equivalent to the condition
$$\ba c
\hbox{$\gamma F$ is satisfied by every interpretation of~$\sigma_2'$}\\
\hbox{that satisfies $\A$, $\gamma(\Ind)$, $\gamma(\Std)$, $\gamma D_0$,
      ~$\gamma D_1$ and $\gamma G$}.
\ea$$
  Consequently $F\lrar G$ is provable in \HTC\ iff
$$\ba c
\hbox{$\gamma F\lrar\gamma G$ is satisfied by every interpretation
      of~$\sigma_2'$}\\
\hbox{that satisfies $\A$, $\gamma(\Ind)$, $\gamma(\Std)$, $\gamma D_0$,
      and~$\gamma D_1$}.
\ea$$
Since all predicate constants occurring in $\Std$ are comparisons,
$\gamma(\Std)$ is equivalent to~$\Std$, so that
$\gamma(\Std)$ here can be replaced by~\Std.  It remains to observe that
$\A$, $\gamma(\Ind)$, \Std, $\gamma D_0$ and~$\gamma D_1$ is the list of
all axioms of $\HTC'$.
\qed

\subsection{Standard HT\nobreakdash-interpretations}\label{sec:standard}

In preparation for the proof of Theorem~\ref{thm:HTCw.strong.equivalence},
we describe here the class of standard HT\nobreakdash-interpretations of the
signature~$\sigma_1$ and prove the soundness and completeness of~$\HTCw$
with respect to standard \HTis.

An interpretation of~$\sigma_1$ is \emph{standard}
if its restriction to~$\sigma_0$ is standard (see Section~\ref{ssec:htsi})
and it satisfies~\Defs.
For every set~$\X$ of precomputed atoms,~$\X^\uparrow$ stands for the
standard interpretation of~$\sigma_1$ defined by the following conditions:
\begin{itemize}
\item[(a)] a precomputed atom is satisfied by~$\X^\uparrow$ iff it belongs
  to~$\X$;
\item[(b)] an extended precomputed atom $\atl^{\bfx;\bfv}_F({\bf v},r)$
  is satisfied by~$\X^\uparrow$ iff
  $$\X^\uparrow\models\left(\exists_{\geq r} \bfx F\right)^\bfv_{\bf v};$$
\item[(c)] an extended precomputed atom $\atm^{\bfx;\bfv}_F({\bf v},r)$
  is satisfied by~$\X^\uparrow$ iff
  $$\X^\uparrow\models\left(\exists_{\leq r} \bfx F\right)^\bfv_{\bf v}.$$
\end{itemize}
The operation $\X\mapsto \X^\uparrow$ is opposite
to the operation \hbox{$I\mapsto I^\downarrow$}
defined in Appendix~\ref{sec:review},
in the sense that
\begin{itemize}
  \item for any standard interpretation~$I$ of~$\sigma_1$,
  $(I^\downarrow)^\uparrow=I$, and
  \item for any set~$\X$ of precomputed atoms, the set of precomputed atoms
  in~$(\X^\uparrow)^\downarrow$ is~$\X$.
\end{itemize}

This construction is extended to \HTis\ as follows.
An \HTi\ $\tuple{\HH,I}$ of~$\sigma_1$ is \emph{standard} if the
restriction of~$I$ to~$\sigma_0$ is standard and $\tuple{\HH,I}$
satisfies~\Defs.  For any standard \HTi\ $\tuple{\HH,I}$, $I$ satisfies
\Defs\ by the persistence property of \HTis\ (Appendix~\ref{sec:review}), so
that~$I$ is standard as well.
For any pair~$\X$,~$\Y$ of sets of precomputed atoms such that
$\X\subseteq \Y$, the pair $\tuple{\X,\Y^\uparrow}$ is an \HTi\
of~$\sigma_1$, because $\X\subseteq \Y \subseteq (\Y^\uparrow)^\downarrow$.
We define \emph{extended precomputed atoms} as atomic formulas $p(\bft)$ over
the signature~$\sigma_1$ such that~$p$ is intensional, and
$\bft$ is a tuple of precomputed terms.
Let~$\HH$ be the superset of~$\X$ obtained from it by adding all
extended precomputed atoms $\atl^{\bfx;\bfv}_F({\bf v},r)$
such that
$$\tuple{\X,\Y^\uparrow}\modelsht
\left(\exists_{\geq r} \bfx F\right)^\bfv_{\bf v}$$
and all extended precomputed atoms  $\atm^{\bfx;\bfv}_F({\bf v},r)$
such that
$$\tuple{\X,\Y^\uparrow}\modelsht
\left(\exists_{\leq r} \bfx F\right)^\bfv_{\bf v}.$$
For every atom $\atl^{\bfx;\bfv}_F({\bf v},r)$ in~$\HH$,
$$\Y^\uparrow\models\left(\exists_{\geq r} \bfx F\right)^\bfv_{\bf v}$$
by persistence.
Consequently every such atom belongs to $(\Y^\uparrow)^\downarrow$.
Similarly, every atom $\atm^{\bfx;\bfv}_F({\bf v},r)$ in~$\HH$
belongs to $(\Y^\uparrow)^\downarrow$ as well.  It follows that
$\tuple{\HH,\Y^\uparrow}$ is an \HTi\ of~$\sigma_1$.  We denote this \HTi\ by
$\tuple{\X,\Y}^\uparrow$.

This \HTi\ is standard.  Indeed, a precomputed atom of the form
$\atl^{\bfx;\bfv}_F({\bf v},r)$ belongs to~$\HH$ iff
the formula $\left(\exists_{\geq r} \bfx F\right)^\bfv_{\bf v}$ is
satisfied by $\tuple{\X,\Y}^\uparrow$, because this \HTi\ interprets
sentences over the signature~$\sigma_0$ in the same way as
$\tuple{\X,\Y^\uparrow}$.  Similarly, $\atm^{\bfx;\bfv}_F({\bf v},r)$
belongs to~$\HH$ iff the formula
$\left(\exists_{\leq r} \bfx F\right)^\bfv_{\bf v}$ is satisfied
by $\tuple{\X,\Y}^\uparrow$.
It follows that $\tuple{\X,\Y}^\uparrow$ satisfies \Defs.

Conversely, every standard \HTi\ of~$\sigma_1$ can be represented in
the form $\tuple{\X,\Y}^\uparrow$.  Indeed, for any standard \HTi\
$\tuple{\HH,I}$ of~$\sigma_1$, take~$\X$ to be the set of precomputed atoms
in~$\HH$, and take~$\Y$ to be~$I^\downarrow$.  Then
$$\X\subseteq\HH\subseteq I^\downarrow=\Y$$
and
$$\tuple{\X,\Y}^\uparrow=\tuple{\HH,\Y^\uparrow}=
\tuple{\HH,(I^\downarrow)^\uparrow}=\tuple{\HH,I}.$$

\medskip\noindent{\bf Soundness and Completeness Theorem.}
\emph{For any set~$\Gamma$ of sentences over~$\sigma_1$ and any
  sentence~$F$ over~$\sigma_1$, $F$ can be derived from~$\Gamma$
  in~$\HTC^\omega$ iff every standard \HTi\ satisfying~$\Gamma$
  satisfies~$F$.
}\medskip

The proof of the theorem refers to \emph{$\omega$-interpretations} of
many-sorted signatures
\cite[Section~5.2]{fan23}.  In case of the signatures~$\sigma_0$,~$\sigma_1$
and~$\sigma_2$, $\omega$-interpretations are characterized by two
conditions:
\begin{itemize}
\item every element of the domain of general variables is represented by
  a precomputed term;
\item every element of the domain of integer variables is represented by
  a numeral.
\end{itemize}
An \emph{$\omega$-model} of a set~$\Gamma$ of sentences is an \HTi\
$\tuple{\HH,I}$ satisfying~$\Gamma$ such that~$I$ is an
$\omega$-interpretation.

\begin{lemma}\label{lem:standard.interpretation}
An~\HTi\ of~$\sigma_1$ is isomorphic to a standard \HTi\ iff it
is an $\omega$-model of~$\Std$ and~$\Defs$.
\end{lemma}

\begin{proof}
  The only-if part is obvious.  If~$\tuple{\HH,I}$ is an
  $\omega$-model of~$\Std$ and~$\Defs$ then the
  function that maps every precomputed term to the corresponding element
  of the domain of general variables in~$I$ is an isomorphism between a
  standard \HTi\ and~$\tuple{\HH,I}$.
\end{proof}

\noindent\emph{Proof of the soundness and completeness theorem.}
The deductive system~\SQHTw\ \cite[Section~5.3]{fan23} for the
signature~$\sigma_1$ can be described as $\SQHT^=$ (see
Appendix~\ref{sec:review}) extended by the two $\omega$-rules from
Section~\ref{sec:omega}.  According to Theorem~4 from that paper,
for any set~$\Gamma$ of sentences over~$\sigma_1$ and any
  sentence~$F$ over~$\sigma_1$, $F$ can be derived from~$\Gamma$
  in~$\SQHTw$ iff every $\omega$-model of~$\Gamma$
  satisfies~$F$.  On the other hand, $\HTCw$ can be described as~$\SQHTw$
  extended by the axioms~\Std\ and~\Defs.  It follows that
for any set~$\Gamma$ of sentences over~$\sigma_1$ and any
  sentence~$F$ over~$\sigma_1$, $F$ can be derived from~$\Gamma$
  in~$\HTCw$ iff every $\omega$-model of~\Std,~\Defs\
  and~$\Gamma$ satisfies~$F$.  The assertion of the theorem follows by
  Lemma~\ref{lem:standard.interpretation}.

\subsection{Grounding}

Recall that in Appendix~\ref{sec:plem}, we subdivided the predicate symbols
of the signature~$\sigma_2$ into extensional and intensional.  This
classification applies, in particular, to the predicate symbols
of~$\sigma_1$: comparison symbols~(\ref{comp}) are extensional, and the
symbols $p/n$, $\atl^{\bfx;\bfv}_F$ and $\atm^{\bfx;\bfv}_F$ are
intensional.

The proof of Theorem~\ref{thm:HTCw.strong.equivalence} refers to the
grounding transformation
$F\mapsto F^{\rm prop}$, which converts sentences over~$\sigma_0$ 
into infinitary propositional combinations of precomputed atoms, and
sentences over~$\sigma_1$ into infinitary propositional
combinations of extended precomputed atoms 
(\citealt[Section~3]{tru12};
\citealt[Section~5]{lif19}; 
\citealt[Section~10]{lif22a};
\citealt[Section~8]{fan22}).%
\footnote{In some of these papers, the transformation $F\mapsto F^{\rm prop}$ is denoted by~$\mathit{gr}$.  
We take
the liberty to identify precomputed terms~$t$ with their names~$t^*$.}
This transformation is defined as follows:
\begin{itemize}
\item if $F$ is $p(\bft)$, where~$p$ is intensional,
  then $F^{\rm prop}$ is obtained from $F$ by replacing each member of~$\bft$
  by the precomputed term obtained from it by evaluating arithmetic
  functions;
  \item if $F$ is $t_1 \prec t_2$, then $F^{\rm prop}$ is $\top$ if the values of $t_1$ and $t_2$ are in the relation~$\prec$,  and $\bot$ otherwise;
  \item $(\neg F)^{\rm prop}$ is $\neg F^{\rm prop}$;
  \item $(F \odot G)^{\rm prop}$ is $F^{\rm prop} \odot G^{\rm prop}$ for every binary connective~$\odot$;
  \item $(\forall X\,F)^{\rm prop}$ is the conjunction of the formulas
    $(F^X_t)^{\rm prop}$
    over all precomputed terms $t$ if $X$ is a general variable, and
    over all numerals $t$ if $X$ is an integer variable;
  \item $(\exists X\,F)^{\rm prop}$ is the disjunction of the formulas
    $(F^X_t)^{\rm prop}$
    over all precomputed terms $t$ if $X$ is a general variable, and
    over all numerals $t$ if $X$ is an integer variable.
  \end{itemize}
  For any set~$\Gamma$ of sentences over~$\sigma_1$, $\Gamma^{\rm prop}$
  stands for $\{F^{\rm prop}\,:\,F\in\Gamma\}$.

  The lemma below relates the meaning of a sentence over~$\sigma_1$
   to the meaning of its grounding.
It is similar to Proposition~4 from
  \citeauthor{tru12}'s article (\citeyear{tru12}) and can be
 proved by induction in a similar way.


\begin{lemma}\label{mirek_prop4}
  For any \HTi\ $\tuple {\HH,I}$ of~$\sigma_1$ such that the restriction
  of~$I$ to~$\sigma_0$ is standard, and any sentence~$F$
  over~$\sigma_1$,
  $$\tuple{\HH,I}\modelsht F\hbox{ iff }
  \tuple{\HH,I^\downarrow}\modelsht F^{\rm prop}.$$
\end{lemma}

The next lemma relates $(\tau^*\Pi)^{\rm prop}$ to~$\tau\Pi$.
   \begin{lemma}\label{prop:taustar-tau}
     For any {\sc mgc} program~$\Pi$,
   every propositional \HTi\ satisfying $\Defs^{\rm prop}$
   satisfies also the formula $(\tau^*\Pi)^{\rm prop}\lrar\tau \Pi$.
\end{lemma}

\begin{proof}
  It is sufficient to consider the case when~$\Pi$
  is a single pure rule~$R$.  The equivalence
$(\tau^*R)^{\rm prop}\lrar\tau R$
is provable in the deductive
system $HT^{\infty}+\Defs^{\rm prop}$ \cite[Theorem~1]{lif22a}.
The assertion of the lemma follows from this theorem, because every
\HTi\ satisfies all axioms
of~$HT^{\infty}$, and satisfaction is preserved by the inference rules of $HT^{\infty}$.
\end{proof}

\subsection{Proof of Theorem~\ref{thm:HTCw.strong.equivalence}}
\label{sec:proof:thm:HTCw.strong.equivalence}

\begin{lemma}\label{lem:grounding.models}
  For any {\sc mgc} program~$\Pi$ and any sets~$\X$,~$\Y$ of
  precomputed atoms such that $\X\subseteq\Y$,
  \beq
  \tuple{\X,\Y}^\uparrow\modelsht\tau^*\Pi
  \eeq{leml}
  iff
$\tuple{\X,\Y}\modelsht\tau\Pi$.
\end{lemma}

\begin{proof}
  Condition~(\ref{leml}) can be  rewritten as
$\tuple{\HH,\Y^\uparrow}\modelsht\tau^*\Pi$,
and, by Lemma~\ref{mirek_prop4}, it is equivalent to
\beq
\tuple{\HH,(\Y^\uparrow)^\downarrow}\modelsht(\tau^*\Pi)^{\rm prop}.
\eeq{lemll}
  On the other hand,
$$\tuple{\HH,\Y^\uparrow}=\tuple{\X,\Y}^\uparrow\modelsht\Defs.$$
  By Lemma~\ref{mirek_prop4}, it follows that
  $\tuple{\HH,(\Y^\uparrow)^\downarrow}$ satisfies $\Defs^{\rm prop}$.
  By Lemma~\ref{prop:taustar-tau}, we can conclude that~(\ref{lemll}) is
equivalent to
$$\tuple{\HH,(\Y^\uparrow)^\downarrow}\modelsht\tau\Pi.$$
This condition is equivalent to $\tuple{\HH,\Y}\modelsht\tau\Pi$, because the
formula~$\tau\Pi$ is formed from precomputed atoms.
\end{proof}

\noindent\emph{Proof of the theorem.}
We need to show that the formula $\tau^*\Pi_1\lrar\tau^*\Pi_2$ is provable
in $\HTCw$ iff $\tau\Pi_1$ is strongly equivalent to~$\tau\Pi_2$.
This formula is provable in $\HTCw$ iff
\beq
\tau^*\Pi_2\hbox{ is derivable in $\HTCw$ from }\tau^*\Pi_1
\eeq{lr7}
and
\beq
\tau^*\Pi_1\hbox{ is derivable in $\HTCw$ from }\tau^*\Pi_2.
\eeq{rl7}
On the other hand, the characterization of strong equivalence of
infinitary propositional formulas in terms of propositional \HTis\
\cite[Theorem~3]{har17} shows that
$\tau\Pi_1$ is strongly equivalent to $\tau\Pi_2$ iff
\beq
\ba c
\hbox{for every propositional \HTi\ $\tuple{\X,\Y}$},\\
\hbox{if $\tuple{\X,\Y}\modelsht\tau\Pi_1$ then }
\tuple{\X,\Y}\modelsht\tau\Pi_2
\ea
\eeq{lr8}
and
\beq
\ba c
\hbox{for every propositional \HTi\ $\tuple{\X,\Y}$},\\
\hbox{if $\tuple{\X,\Y}\modelsht\tau\Pi_2$ then }
\tuple{\X,\Y}\modelsht\tau\Pi_1
\ea
\eeq{rl8}
We will show that conditions~(\ref{lr7}) and~(\ref{lr8}) are
equivalent to each other; the equivalence between~(\ref{rl7})
and~(\ref{rl8}) is proved in a similar way.

Assume that condition~(\ref{lr7}) is satisfied but condition~(\ref{lr8})
is not, so that $\tuple{\X,\Y}\modelsht\tau\Pi_1$ and
$\tuple{\X,\Y}\not\modelsht\tau\Pi_2$ for some propositional \HTi\
$\tuple{\X,\Y}$.  By Lemma~\ref{lem:grounding.models},
$\tuple{\X,\Y}^\uparrow\modelsht\tau^*\Pi_1$ and
$\tuple{\X,\Y}^\uparrow\not\modelsht\tau^*\Pi_2$.
Thus there exists a standard \HTi\ of~$\sigma_1$ that satisfies
$\tau^*\Pi_1$ but does not satisfy $\tau^*\Pi_2$.  This is in contradiction
with the fact that $\HTCw$ is sound with respect to standard interpretations
(Appendix~\ref{sec:standard}).

Assume now that (\ref{lr7}) is not satisfied.  Since
$\HTCw$ is complete with respect to standard interpretations
(Appendix~\ref{sec:standard}), there exists a standard interpretation
that satisfies $\tau^*\Pi_1$ but does not satisfy $\tau^*\Pi_2$.  Consider
a representation of this interpretation in the form $\tuple{\X,\Y}^\uparrow$.
By Lemma~\ref{lem:grounding.models}, $\tuple{\X,\Y}\modelsht\tau\Pi_1$
and $\tuple{\X,\Y}\not\modelsht\tau\Pi_2$, so that condition~({\ref{lr8}) is
not satisfied either.

\subsection{Proofs of Theorems~\ref{thm:j:HTC.strong.equivalence}
  and~\ref{thm:stronger}}
\label{sec:proof:thm:j:HTC.strong.equivalence}
\label{sec:proof:thm:stronger}

\subsubsection{Deductive system \HTCww} \label{sec:htcw2}

Proofs of Theorems~\ref{thm:j:HTC.strong.equivalence}
and~\ref{thm:stronger} refer to the deductive system $\HTCww$, which is a
straightforward extension of~$\HTCw$ to the signature~$\sigma_2$.
Its derivable objects are sequents over~$\sigma_2$.
Its axioms and inference rules are those of intuitionistic logic for the
signature~$\sigma_2$ extended by
\begin{itemize}
\item  axiom schemas (\ref{hosoi}) and (\ref{sqht})
for all formulas $F$, $G$, $H$ over~$\sigma_2$,
\item axioms \Std\  and \Defs, and
\item the $\omega$-rules from Section~\ref{sec:omega} extended to
  sequents over~$\sigma_2$.
\end{itemize}

Any sentence provable in~\HTC\ can be derived in~$\HTCww$
from~$D_0$ and~$D_1$.
Indeed, the only axioms of~\HTC\ that
are not included in~$\HTCww$ are~\Ind,~$D_0$ and~$D_1$, and all
instances of~\Ind\ can be proved using the second $\omega$-rule,
as discussed in Section~\ref{sec:omega}.
We will prove a stronger assertion:

\begin{lemma}\label{thm:HTCws=>newdefs}
Any sentence provable in~\HTC\ can be derived in~$\HTCww$
from~$D_0$.
\end{lemma}

We will prove also the following conservative extension property:

\begin{lemma}\label{lem:conservative.extension.sigma2}
Every sentence over the signature~$\sigma_1$ derivable
in~$\HTCww$ from~$D_0$ is provable in~$\HTCw$.
\end{lemma}

The assertion of Theorem~\ref{thm:stronger} follows from these two lemmas.

The assertion of Theorem~\ref{thm:j:HTC.strong.equivalence} follows from
Theorems~\ref{thm:HTCw.strong.equivalence} and~\ref{thm:stronger}.

\subsubsection{Some formulas derivable in \HTCw and \HTCww}

In Appendix~\ref{ssec:continued}, we showed that formula~\eqref{htap:explus1} is
provable in~\HTC.  All axioms of~\HTC\ used in that proof are among the
axioms of $\HTCw$, so that this formula is provable
in~$\HTCw$ as well.

\medskip\noindent\emph{Claim:}
For any formula~$F$ over~$\sigma_0$ and any integers~$m$,~$n$ such that
$m\geq n$, the formula
\begin{gather}
  \exists_{\geq \num m} \bfu \, F \to \exists_{\geq \num n} \bfu \, F
  \label{htap:exists.geq.succ}
\end{gather}
is provable in~$\HTCw$. 

\begin{proof}
  It is sufficient to consider the case when $m=n+1$; then the
  general case will follow by induction.  We can also assume
  that~$n$ is positive, because otherwise the consequent
  of~(\ref{htap:exists.geq.succ}) is~$\top$.
From~\eqref{htap:explus1},
\begin{gather*}
\exists_{\geq \num{n+1}}\bfu\,F \to
\exists\bfw(\exists_{\geq \num n}\bfu(\bfw<\bfu \land F)).
\end{gather*}
We can rewrite the consequent of this implication as
\begin{gather*}
  \exists \bfw \bfu_1 \dotsm \bfu_n \left(
    \bigwedge_{i=1}^n \left(\bfw < \bfu_i \land F^\bfu_{\bfu_1} \right)
    \land \bigwedge_{i<j} \neg (\bfu_i = \bfu_j)
    \right).
\end{gather*}
It implies
\begin{gather*}
  \exists \bfu_1 \dotsm \bfu_n \left(
    \bigwedge_{i=1}^n F^\bfu_{\bfu_1} 
    \land \bigwedge_{i<j} \neg (\bfu_i = \bfu_j)
    \right),
\end{gather*}
which is the consequent of~\eqref{htap:exists.geq.succ}.
\end{proof}

\noindent\emph{Claim:}
For any formula~$F$ over~$\sigma_0$, any integer~$n$, and any
precomputed term~$r$, the formula
\begin{gather}
  \num n \geq r \wedge \exists_{\geq \num n} \bfu \, F \to \exists_{\geq r} \bfu \, F
  \label{htap:exists.geq.monotonicity}
\end{gather}
is provable in~$\HTCw$. 
\begin{proof}
  If~$\num n < r$ then the antecedent of~\eqref{htap:exists.geq.monotonicity}
  is equivalent to~$\bot$.
  If~$r \leq \num 0$ then the consequent
  of~\eqref{htap:exists.geq.monotonicity} is~$\top$.
  If $\num n\geq r>0$ then~$r$ is a numeral~$\num m$, because the set of
  numerals is contiguous, so that~\eqref{htap:exists.geq.monotonicity}
  follows from~\eqref{htap:exists.geq.succ}.
\end{proof}

In Appendix~\ref{ssec:continued}, we showed that formula~\eqref{htap:exex}
is provable in~\HTC.  All axioms of~\HTC\ used in that proof are among the
axioms of $\HTCw$, so that this formula is provable
in~$\HTCw$ as well.

\medskip\noindent\emph{Claim:}
The formula
\beq
\forall \bfv N(N\geq \num 0 \to
(\atm^{\bfx;\bfv}_F(\bfv,N)
\lrar \neg\atl^{\bfx;\bfv}_F(\bfv,N+\num 1)))
\eeq{neweq}
is provable in~$\HTCww$.

\medskip\noindent\emph{Proof:}
By~\emph{Defs} and~\eqref{htap:exex}, for every nonnegative~$n$,
$$
    \atm^{\bfx;\bfv}_F(\bfv,\num n)
    \lrar
    \exists_{\leq \num n} \bfu \, F
    \lrar
    \neg\exists_{\geq \num {n+1}} \bfu \, F
    \lrar
    \neg\atl^{\bfx;\bfv}_F(\bfv,\num n+\num 1).
    $$
    It follows that for every integer~$n$,
$$\forall \bfv (\num n\geq \num 0 \to
(\atm^{\bfx;\bfv}_F(\bfv,\num n)
\lrar \neg\atl^{\bfx;\bfv}_F(\bfv,\num n+\num 1))).$$
Formula~(\ref{neweq}) follows by the second $\omega$-rule.
\qed

\medskip\noindent\emph{Claim:}
The formula
\begin{gather}
  \forall \bfv Y \left( \exists N \left(
    N \geq Y \wedge
    \atl^{\bfx;\bfv}_F(\bfv,N)
    \right)
    \lrar
    \atl^{\bfx;\bfv}_F(\bfv,Y)
    \right)
  \label{htcw:Atleast.monotonicity}
\end{gather}
is provable in~$\HTCww$.

\begin{proof}
  \emph{Left-to-right:} take any integer~$n$ and precomputed term~$r$.
  By~\eqref{def1}, the universal closure
  of~\eqref{htap:exists.geq.monotonicity} can be rewritten as
  $$  \forall \bfv \left(
    \num n \geq r \wedge
    \atl^{\bfx;\bfv}_F(\bfv,\num n)
    \to
    \atl^{\bfx;\bfv}_F(\bfv,r)
  \right).$$
  By the $\omega$-rules, it follows that
  $$  \forall N Y \bfv \left(
    N \geq Y \wedge
    \atl^{\bfx;\bfv}_F(\bfv,N)
    \to
    \atl^{\bfx;\bfv}_F(\bfv,Y)
  \right),$$
  which is equivalent to the implication to be proved.
  
\emph{Right-to-left:}  We will show that
  \beq
  \forall \bfv \left(    \atl^{\bfx;\bfv}_F(\bfv,r)\to
    \exists N \left(
    N \geq r \wedge
    \atl^{\bfx;\bfv}_F(\bfv,N)
    \right) \right)
\eeq{r2l}
    for every precomputed term~$r$; then the implication to be proved will
    follow by the second $\omega$-rule.
    Since the set of numerals is contiguous, three cases are possible:
    (1)~$r<\num n$ for all integers~$n$;
    (2)~$r$ is a numeral;
    (3)~$r>\num n$ for all integers~$n$.
    In the last case, the antecedent of~(\ref{r2l}) is equivalent to~$\bot$
    by~\eqref{def1}.  Otherwise, assume $\atl^{\bfx;\bfv}_F(\bfv,r)$; we
    need to find~$N$ such that $N \geq r$ and $\atl^{\bfx;\bfv}_F(\bfv,N)$.
    If $r<\num n$ for all~$n$ then take $N=\num 0$;
    $\atl^{\bfx;\bfv}_F(\bfv,\num 0)$ follows from~(\ref{def1}).  If~$r$ is
    a numeral then take $N=r$.
\end{proof}

In Appendix~\ref{ssec:continued} we showed that
formula~\eqref{htcn:st<->defs} is
provable in~\HTC.  The axioms~$D_1$ are not used in that proof.  It follows
that formula~\eqref{htcn:st<->defs} is derivable in~$\HTCww$ from~$D_0$.

%
 
 \medskip\noindent\noindent\emph{Claim:} The sentence
 \begin{gather}
   \forall \bfv N \left(
   \exists \bfx 
     \st^{\bfx,\bfv}_F(\bfx,\bfv, N)
     \lrar
     \atl^{\bfx;\bfv}_F(\bfv,N)
     \right)
   \label{eq:st<->Atleast*}
 \end{gather}
is derivable in~$\HTCww$ from~$D_0$.
 
\begin{proof}
  For every integer~$n$, the sentence
 \begin{gather*}
 \forall\bfv(  \exists \bfx 
     \st^{\bfx,\bfv}_F(\bfx,\bfv, \num n)
     \lrar
     \atl^{\bfx;\bfv}_F(\bfv,\num n))
 \end{gather*}
 is derivable in~$\HTCww$ from~$D_0$, because it follows from~\eqref{def1}
 and~\eqref{htcn:st<->defs}.  Then~(\ref{eq:st<->Atleast*}) follows by the
 second $\omega$\nobreakdash-rule.
\end{proof}

%

\subsubsection{Proof of Lemma \ref{thm:HTCws=>newdefs}}

To prove Lemma~\ref{thm:HTCws=>newdefs}}, we need to show that all
instances of~$D_1$ can be derived in~$\HTCww$ from~$D_0$.

\medskip\noindent\emph{Proof of~(\ref{d1a}).}
By~\eqref{htcw:Atleast.monotonicity}, $\atl^{\bfx;\bfv}_F(\bfv,Y)$ is
equivalent to
$$
\exists N(N \geq Y \land\atl^{\bfx;\bfv}_F(\bfv,N)).
$$
By~\eqref{eq:st<->Atleast*}, this formula is equivalent to
$$\exists N( N \geq Y \land \exists \bfx \st^{\bfx;\bfv}_F(\bfx,\bfv,N)),$$
which can be further rewritten as
$$\exists \bfx N(\st^{\bfx;\bfv}_F(\bfx,\bfv,N)\land N\geq Y).$$

\medskip\noindent\emph{Proof of~(\ref{d1b}).}
We will prove the equivalence
\begin{gather}
\atm^{\bfx;\bfv}_F(\bfv,Y)
\lrar
\forall \bfx N(\st^{\bfx;\bfv}_F(\bfx,\bfv,N)\to N\leq Y)
\label{eq:2:st<->Atleast**}
\end{gather}
by cases, using the \Std\ axiom
\begin{gather*}
  Y < \num 0\,\vee\,\forall M (M < Y)\,\vee\,
  \exists M (M = Y \wedge M \geq \num 0).
\end{gather*}

\medskip\noindent\emph{Case~1:} $Y < \num 0$.
By~\emph{Defs}, the left-hand side of~\eqref{eq:2:st<->Atleast**} is equivalent to~$\bot$.
Furthermore, by~$D_0$, $\st^{\bfx,\bfv}_F(\bfu,\bfv, 0)$ and thus the
right-hand side of~\eqref{eq:2:st<->Atleast**} is equivalent to~$\bot$ as well.

\medskip\noindent\emph{Case~2:} $\forall M (M < Y)$.
By~\emph{Defs}, the left-hand side of~\eqref{eq:2:st<->Atleast**} is
equivalent to~$\top$.  The right-hand side is equivalent to~$\top$ as well.

\medskip\noindent\emph{Case~3:} $M = Y$ and $M \geq 0$.
Formula~(\ref{eq:2:st<->Atleast**}) can be rewritten as
\beq
\atm^{\bfx;\bfv}_F(\bfv,M)
\lrar
\forall \bfx N(\st^{\bfx;\bfv}_F(\bfx,\bfv,N)\to N\leq M).
\eeq{xx}
By~\eqref{neweq}, the left-hand side is equivalent to
$\neg \atl^{\bfx;\bfv}_F(\bfv,M+\num 1)$.  Hence, by~(\ref{d1a}), it is
equivalent to
$$
\neg \exists \bfx N(\st^{\bfx;\bfv}_F(\bfx,\bfv,N+\num 1)
\land N+\num 1\geq M+\num 1)
$$
and furthermore to
$$
\neg \exists \bfx N(\st^{\bfx;\bfv}_F(\bfx,\bfv,N)
\land N\geq M+\num 1).
$$
This formula can be further rewritten as
$$
\forall \bfx N(\st^{\bfx;\bfv}_F(\bfx,\bfv,N)
\to \neg(N\geq M+\num 1)),
$$
which is equivalent to the right-hand side of~(\ref{xx}).

\subsubsection{Proof of Lemma~\ref{lem:conservative.extension.sigma2}}
  
Ley~$F$ be a sentence over the signature~$\sigma_1$ that is  derivable
in~$\HTCww$ from~$D_0$.  We will show that every standard \HTi\ of~$\sigma_1$
satisfies~$F$; then the provability of~$F$ in~$\HTCw$ will follow from the
completeness of~$\HTCw$ (Appendix~\ref{sec:standard}).

Consider a standard \HTi\ $\tuple{\HH,I}$ of~$\sigma_1$.
Let~$I'$ be the extension of~$I$ to the signature~$\sigma_2$ defined by the
condition: an extended precomputed atom
$\st^{\bfx,\bfv}_F(\bfxx,\bfvv,\num n)$ is satisfied by~$I'$ iff $n\leq 0$ or
  \begin{itemize}
  \item $n>0$,
  \item $I \models F^{\bfx,\bfv}_{\bfxx,\bfvv}$, and
  \item there exist at least~$n$ tuples~$\bfyy$ of
precomputed terms such that $\bfyy\geq \bfxx$ and
$I \models F^{\bfx,\bfv}_{\bfyy,\bfvv}$.
\end{itemize}
Since~$\tuple{\HH,I}$ is standard, $I'$ is an $\omega$-iterpretation.
Furthermore, let~$\HH'$ be the set of extended precomputed atoms obtained
from~$\HH$ by
adding the atoms $\st^{\bfx,\bfv}_F(\bfxx,\bfvv,\num n)$ such that $n\leq 0$ or
  \begin{itemize}
  \item $n>0$,
  \item $\tuple{\HH,I} \modelsht F^{\bfx,\bfv}_{\bfxx,\bfvv}$, and
  \item there exist at least~$n$ tuples~$\bfyy$ of
precomputed terms such that $\bfyy\geq \bfxx$ and
$\tuple{\HH,I} \modelsht F^{\bfx,\bfv}_{\bfyy,\bfvv}$.
\end{itemize}
From the persistence property of \HTis\ (Appendix~\ref{sec:review}) we can
conclude that each of the atoms added to~$\HH$ is satisfied by~$I'$.
Hence $\tuple{\HH',I'}$ is an \HTi\ of~$\sigma_2$.

We will show that
\beq
\tuple{\HH',I'}\hbox{ satisfies \Std,~\Defs, and }D_0.
\eeq{xyz}
Then the assertion of the lemma will follow.  Indeed,
the deductive system~$\HTCww$ can be described as $\SQHTw$
(see Appendix~\ref{sec:standard}) over~$\sigma_2$ extended by the
axioms~\Std\ and~\Defs.  Hence~$F$ is derivable in $\SQHTw$
over~$\sigma_2$ from~\Std,~\Defs\ and~$D_0$.  By~(\ref{xyz}),
$\tuple{\HH',I'}$ is an $\omega$-model of these sentences.  By the
soundness of~$\SQHTw$ \cite[Theorem~4]{fan23}, it follows that~$F$
is satisfied by~$\tuple{\HH',I'}$.  Since~$F$ is a sentence over~$\sigma_1$,
we conclude that~$F$ is satisfied by~$\tuple{\HH,I}$.

\medskip\noindent\emph{Proof of~(\ref{xyz}):}
\medskip

For $\Defs$ and~$\Std$, this assertion follows from the fact that
$\tuple{\HH',I'}$ extends the interpretation
$\tuple{\HH,I}$ of~$\sigma_1$, which satisfies $\Defs$ and~$\Std$ because
it is standard.

For~$D_0$, consider the more difficult axiom schema in this group, the
last one.  We need to check that
for any tuples~$\bfxx$ and~$\bfvv$ of precomputed terms and any positive~$n$,
$\tuple{\HH',I'}$ satisfies
\beq
\st^{\bfx;\bfv}_F(\bfxx,\bfvv,\num{n+1})
  \lrar F^{\bfx,\bfv}_{\bfxx,\bfvv} \land \exists \bfu(\bfxx < \bfu \land \st^{\bfx,\bfv}_F(\bfu,\bfvv,\num n)).
\eeq{eqg}
We need to check, in other words, that~$I'$ satisfies the left-hand side
of~(\ref{eqg}) iff~$I'$ satisfies the right-hand side, and similarly for
$\tuple{\HH',I'}$.

Assume that~$I'$ satisfies the left-hand side. Then
$I \models F^{\bfx,\bfv}_{\bfxx,\bfvv}$, and
\begin{center}
  there exist at least~$n+1$ tuples~$\bfyy$ such
  that $\bfyy\geq \bfxx$ and $I\models F^{\bfx,\bfv}_{\bfyy,\bfvv}$.
 \end{center}
It follows that
\begin{center}
  there exist at least~$n$ tuples~$\bfyy$ such
  that $\bfyy>\bfxx$ and $I\models F^{\bfx,\bfv}_{\bfyy,\bfvv}$.
\end{center}
Pick such a group of $n$ tuples, and let~$\bfuu$ be the least among them.
Then $\bfuu>\bfxx$, and
\begin{center}
  there exist at least~$n$ tuples~$\bfyy$ such
  that $\bfyy\geq\bfuu$ and $I\models F^{\bfx,\bfv}_{\bfyy,\bfvv}$.
\end{center}
It follows that~$I$ satisfies
$\bfxx < \bfuu \land \st^{\bfx,\bfv}_F(\bfuu,\bfvv,\num n)$,
and consequently satisfies the right-hand side of~(\ref{eqg}).

Assume now that~$I'$ satisfies the right-hand side of~(\ref{eqg}).  Then
$I'\models F^{\bfx,\bfv}_{\bfxx,\bfvv}$, and there exists a
tuple~$\bfuu$ such that $\bfxx<\bfuu$ and
$I'\models \st^{\bfx,\bfv}_F(\bfuu,\bfvv,\num n)$.
Hence $I\models F^{\bfx,\bfv}_{\bfxx,\bfvv}$, and
\begin{center}
  there exist at least~$n$ tuples~$\bfyy$ such that $\bfyy\geq\bfuu$ and
  $I\models F^{\bfx,\bfv}_{\bfyy,\bfvv}$.
\end{center}
It follows that
\begin{center}
  there exist at least~$n+1$ tuples~$\bfyy$ such that $\bfyy\geq\bfxx$ and
  $I\models F^{\bfx,\bfv}_{\bfyy,\bfvv}$,
\end{center}
so that~$I'$ satisfies the left-hand side of~(\ref{eqg}).

For the \HTi\ $\tuple{\HH,I}$ the reasoning is similar.

\newpage

Competing interests: The authors declare none

\end{document}